\documentclass[journal,twocolumn,10pt]{IEEEtran}
\IEEEoverridecommandlockouts

\ifCLASSINFOpdf
\usepackage[pdftex]{graphicx}
\else
\fi

\normalsize
\usepackage{amsfonts}
\usepackage[caption=false,font=footnotesize]{subfig}
\usepackage{amsmath,amssymb}
\usepackage{verbatim}
\usepackage{setspace}
\usepackage{bbm}
\usepackage{cite}
\usepackage{epstopdf}
\usepackage{balance}
\usepackage{amsthm} 
\newtheorem{theorem}{Theorem}
\newtheorem{corollary}{Corollary}
\newtheorem{proposition}{Proposition}
\newtheorem{lemma}{Lemma}
\newtheorem{remarks}{Remark}

\usepackage{hyperref}
\newcounter{mytempeqncnt}

\usepackage{bbding}
\begin{document}
%
\title{A Stochastic Geometry Analysis of Large-scale Cooperative Wireless Networks Powered by\\ Energy Harvesting}

\author{Talha~Ahmed~Khan, Philip~Orlik, Kyeong~Jin~Kim, Robert~W.~Heath~Jr., and~Kentaro~Sawa
\thanks{T. A. Khan and R. W. Heath Jr. are with The University of Texas at Austin, Austin, TX, USA (Email: \{talhakhan, rheath\}@utexas.edu). %
{P. Orlik and K. J. Kim are with Mitsubishi Electric Research Lab (MERL), Cambridge, MA, USA (Email:\{porlik, kkim\}@merl.com). K. Sawa is with Mitsubishi Electric Corporation IT R\&D Center, Kamakura, Kanagawa, Japan (Email: Sawa.Kentaro@bk.MitsubishiElectric.co.jp).}}
\thanks{This work was initiated while T. A. Khan was with MERL, and was supported by a gift from MERL.} %
\thanks{Parts of this paper were presented at the 2015 IEEE International Conference on Communications \cite{talha2015}.}}
\maketitle
\begin{abstract}
Energy harvesting is a technology for enabling green, sustainable, and autonomous wireless networks. In this paper, a large-scale wireless network with energy harvesting transmitters is considered, where a group of transmitters forms a cluster to cooperatively serve a desired receiver amid interference and noise. 
To characterize the link-level performance, closed-form expressions are derived for the transmission success probability at a receiver in terms of key parameters such as node densities, energy harvesting parameters, channel parameters, and cluster size, for a given cluster geometry.
The analysis is further extended to characterize a network-level performance metric, capturing the 
tradeoff between link quality and the fraction of receivers served.  
Numerical simulations validate the accuracy of the analytical model.
Several useful insights are provided. For example, while more cooperation helps improve the link-level performance, the network-level performance might degrade with the cluster size. 
Numerical results show that a small cluster size (typically 3 or smaller) optimizes 
the network-level performance. 
Furthermore, substantial performance can be extracted with a relatively small energy buffer. Moreover, the utility of having a large energy buffer increases with the energy harvesting rate as well as with the cluster size in sufficiently dense networks.

\end{abstract}


%
\IEEEpeerreviewmaketitle
\section{Introduction}
Energy harvesting is a promising approach for realizing self-powered wireless networks.
A wireless device equipped with energy harvesting capability may extract energy from natural or man-made sources such as solar radiations, wind, radio frequency (RF) signals, indoor lighting, etc. \cite{EHreview2015}. Energy harvesting could potentially transform both infrastructure-based as well as ad hoc wireless networks. For instance, in cellular systems, energy harvesting could help reduce the operating expenditures for the cell-sites, reduce the carbon footprint as well as facilitate cell-site deployment \cite{dhillon2014fund}. Similarly, energy harvesting is also closely related to the Internet of Things \cite{IoT2014}, which broadly is a network consisting of everyday objects such as machines, buildings, vehicles, etc. Many of these \emph{smart} objects will contain low-power wireless sensors that communicate with other devices and/or a cloud/control unit. Energy harvesting can potentially enhance the battery lifetimes while simplifying the network maintenance (for instance, with an energy harvesting device, no human intervention would be needed for battery replacement), thus providing the much-needed autonomy for sustaining such networks\cite{IoT2014,kulkarni2011ehsurvey,EHreview2015}.        
Energy harvesting devices need new communication protocols.
Due to limited energy storage capacity and depending on the type of harvesting, the energy availability at the device varies over time. This leads to a model where energy arrivals are bursty. Several papers have proposed optimal transmission policies assuming causal or non-causal knowledge about energy arrivals for different setups (see \cite{EHreview2015,gunduz2014designing} for a comprehensive review). For example, a point-point link \cite{ozel2011transmission,zhang2012eh}, an interference channel \cite{kaya2012int}, and a broadcast channel \cite{broadcast2011energy} have been considered. 
While prior research has mostly considered simple information-theoretic setups, some recent studies have 
investigated the network-level dynamics in large \textit{non-cooperative} wireless networks powered by energy harvesting \cite{huang2013Spatial,vaze2013gsip,huang2013cog,dhillon2014fund}.

Stochastic geometry is emerging as a popular tool for analyzing a variety of setups ranging from ad hoc, to cognitive and cellular networks. It often leads to tractable analytical models that yield general performance insights, thus obviating the need of exhaustive simulations \cite{haenggi2012stochastic}. 
The performance of ad hoc networks was characterized using metrics such as outage probability and transmission capacity \cite{baccelli2006aloha,weber2010tcap,elsawy2013survey}.
Similar analysis has been applied to single and multi-tier cellular networks under different assumptions about cell association, scheduling and power control \cite{elsawy2013survey,andrews2011tractable,dhillon2012modeling}. Multi-cell cooperation has been analyzed for different cooperation models in \cite{lee2015base,akoum2013interference,baccelli2014pairwise,nigam2013coordinated,tanbourgi2013tractable}. 
For example, dynamic coordinated beamforming was treated in \cite{lee2015base}, random clustering with intercell interference nulling was considered in \cite{akoum2013interference}, and pairwise cooperation with limited channel knowledge was analyzed in \cite{baccelli2014pairwise}. Similarly, joint transmission without prior channel knowledge and/or tight synchronization has also been considered \cite{nigam2013coordinated,tanbourgi2013tractable}.
None of the aforementioned work\cite{lee2015base,akoum2013interference,baccelli2014pairwise,nigam2013coordinated,tanbourgi2013tractable} on \textit{cooperative} networks, however, considers energy harvesting.

Stochastic geometry has been used to analyze energy harvesting systems. Large-scale self-powered ad hoc networks have been analyzed in \cite{huang2013Spatial} and \cite{vaze2013gsip}. In \cite{huang2013Spatial}, the network model consists of a large number of energy harvesting transmitters, where each transmitter has a dedicated receiver located a fixed distance away. Leveraging tools from stochastic geometry and random walk theory, spatial throughput was derived by optimizing over the transmission power. For a similar setup, the authors in \cite{vaze2013gsip} derived the transmission capacity for a random access network by optimizing over the medium access probability. 
Self-powered heterogeneous cellular networks have been considered in \cite{dhillon2014fund}. In \cite{dhillon2014fund}, base-station availability (i.e., the fraction of the time it can remain ON) was analytically characterized using tools from random walk theory and stochastic geometry.  
The work in \cite{dhillon2014fund,huang2013Spatial} and \cite{vaze2013gsip}, however, does not consider any node cooperation or joint transmission at the physical layer. 
Cooperative/joint transmission seems particularly attractive for energy harvesting networks, as it could compensate for the performance loss due to uncertain energy availability at the transmitters.


In this paper, we consider a large-scale network of transmitters and receivers, where a receiver node is jointly served by a cluster consisting of its $K$ closest self-powered transmitter nodes.
This model is attractive for many scenarios involving energy harvesting wireless communications such as self-powered sensor networks, self-powered wireless hotspots, and other IoT-inspired applications of the future\cite{IoT2014,EHreview2015,kulkarni2011ehsurvey}.
We provide a tractable framework to characterize the system performance as a function of key parameters such as the cluster size, the energy harvesting capability, the transmitter/receiver densities and other network and channel parameters.  
We model the locations of the transmitters and receivers using independent Poisson point processes (PPPs). 
To reap the benefits of cooperation, the transmitters are grouped into clusters such that all the in-cluster transmitters jointly serve a common receiver, which is subjected to interference from the out-of-cluster nodes. Channel acquisition and node coordination, which is formidable even for conventional networks, is typically exacerbated with energy harvesting nodes. This motivates us to adopt non-coherent joint transmission as the cooperation model.   
The performance of such a \textit{cooperative} self-powered wireless network in a stochastic geometry framework has not been analyzed in the literature.


The proposed analytical model captures the key interplay between the cluster size and the transmitter and receiver densities.
Note that a transmitter cluster may have multiple candidate receivers, only one of which will be served in a given resource. We therefore consider a performance metric that captures the two key events influencing the overall performance: (i) a receiver is selected for service (modeled via cluster access probability in Section \ref{secClus}), and (ii) the transmission is successful (modeled via link success probability in Section \ref{secLink}).
For the former, we propose an analytical approximation for the cluster access probability in terms of the cluster size and the ratio of the transmitter and receiver densities.
For the latter, we derive simple analytical expressions that characterize the link performance as a function of system parameters (e.g., energy harvesting rate, energy buffer size, transmitter density), channel parameters and cluster geometry, while accounting for the heterogeneous network interference. Leveraging these results (each being a novel contribution in itself), a closed-form analytical expression is derived for the overall performance metric, and validated using simulations.

We also investigate the impact of cluster size, energy harvesting rate and energy buffer size on the overall performance. Our findings suggest that (i) there is an optimal cluster size that maximizes the overall performance given the density parameters; (ii) the optimal cluster size increases with the ratio of transmitter and receiver densities and typically ranges from 1 to 3; (iii) a relatively small energy buffer size (typically large enough to store 10 or fewer transmissions in the considered setup) is sufficient for extracting performance gains; and (iv) the utility of having a large energy buffer increases with the energy harvesting rate as well as with the cluster size when the density ratio is sufficiently large. 
Our analytical model is applicable to a general class of networks, with the traditionally-powered cooperative and non-cooperative networks as special cases. 

The rest of the paper is organized as follows. The system model is described in Section \ref{secSYS}. Using tools from stochastic geometry, the analytical expressions for the considered performance metrics are derived in Section \ref{secSTOCH}. Section \ref{secSIM} presents the simulation results and Section \ref{secCONCLUSION} concludes the paper.
\section{System Model}\label{secSYS}
We now describe the energy harvesting model, the underlying assumptions about the considered network, and the cooperation scenario.
\subsection{Energy Harvesting Model}\label{secEH}
We consider a large wireless network consisting of transmitters that are equipped with energy harvesting modules (e.g., RF energy harvesting). The energy arrivals are assumed to be random and independent across nodes. None of the transmitters are privy to any non-causal information about energy arrivals.
We now describe the energy harvesting model for an arbitrary transmitter equipped with an energy buffer of size $S\in\mathbb{N}$. 
The energy arrives at the buffer with rate $\rho$ following an independent and identically distributed (IID) Bernoulli process\footnote{
Due to analytical tractability, this is a common approach for modeling an energy harvesting process (e.g., see\cite{jeon2015stability,vaze2013gsip,ibrahim2015stability}).
Conceptually, the energy harvesting rate in this model approximates the average energy arrival rate of an actual (continuous) energy harvesting process.  
}, i.e., with probability $\rho$, one unit of energy arrives at the buffer in time-slot $t$, while $1-\rho$ is the probability that no energy arrives at the buffer in that slot. A node may choose to transmit with fixed power $P$ if it has sufficient energy in the buffer. No power control is assumed, therefore each transmission depletes the buffer of $P$ units of energy. The energy arrivals are modeled using a birth-death Markov process along the lines of \cite{huang2013Spatial,vaze2013gsip,jeon2015stability,ibrahim2015stability}. 

For medium access, we consider a slotted ALOHA based random access protocol where in each time-slot, a node (having sufficient energy) accesses the medium with probability $p_{\text{ch}}$ independently of other nodes. 
Let $p_S$ denote the probability that a node has the requisite amount of energy available in the buffer of size $S$.
Then, $p_S=\Pr\left[A_S(t)\geq P\right]$, where $A_S(t)$ denotes the state (i.e., energy level) of the buffer at time $t$. 
We now define $p_{\text{tr}}$, the transmission probability of an arbitrary node, and express it as a function of system parameters. 

\begin{lemma}\normalfont\label{lem1}
	For energy arrivals with rate $\rho >0$, finite energy buffer of size $S\in\mathbb{N}$, and channel access probability $p_{\text{ch}}>0$, it follows that $p_{\text{tr}}=p_{\text{ch}} p_S$ in steady state, where 
	\begin{align}
		p_S &=
		\begin{cases}
			\frac{\rho}{\rho+p_{\text{ch}}-\rho p_{\text{ch}}} & S=1 \\
			\frac{\frac{\rho}{p_{\text{ch}}}\left(1-\left(\frac{\rho\left(1-p_{\text{ch}}\right)}{p_{\text{ch}}\left(1-\rho\right)}\right)^{S}\right)}{1-\frac{\rho}{p_{\text{ch}}}\left(\frac{\rho\left(1-p_{\text{ch}}\right)}{p_{\text{ch}}\left(1-\rho\right)}\right)^{S}} & {S>1}, \rho\neq p_{\text{ch}}  \\
			\frac{S}{S+1-\rho} & {S>1}, \rho=p_{\text{ch}}
		\end{cases}
	\end{align}
for the case $P=1$.	
\end{lemma}
\begin{proof}
	It can be proved by solving the balance equations for the resulting Markov chain (e.g., see \cite{huang2013Spatial}).
\end{proof}
Note that $0<p_{\text{tr}}\leq p_{\text{ch}}$ since $\lim\limits_{\rho\rightarrow 1}p_{\text{tr}}=p_{\text{ch}}$, where $\rho=1$ corresponds to the case when the node is powered by conventional power sources. Furthermore, $p_{\text{ch}}$ is fixed throughout the network. Therefore, the transmission probability of a node varies as a function of the energy harvesting rate and buffer size. In other words, the higher the $p_{\text{tr}}$ of a node, the more superior the energy harvesting capability (i.e., harvesting rate and/or buffer size). Note that Lemma 1 has been specialized for the case $P=1$ for simplicity. For other values of $P$, it is possible to solve the balance equations of the corresponding Markov chain to calculate $p_S$. 
\subsection{Network Model} 
\begin{figure}[t]
	\centering
	\framebox{\includegraphics[width=3.0in]{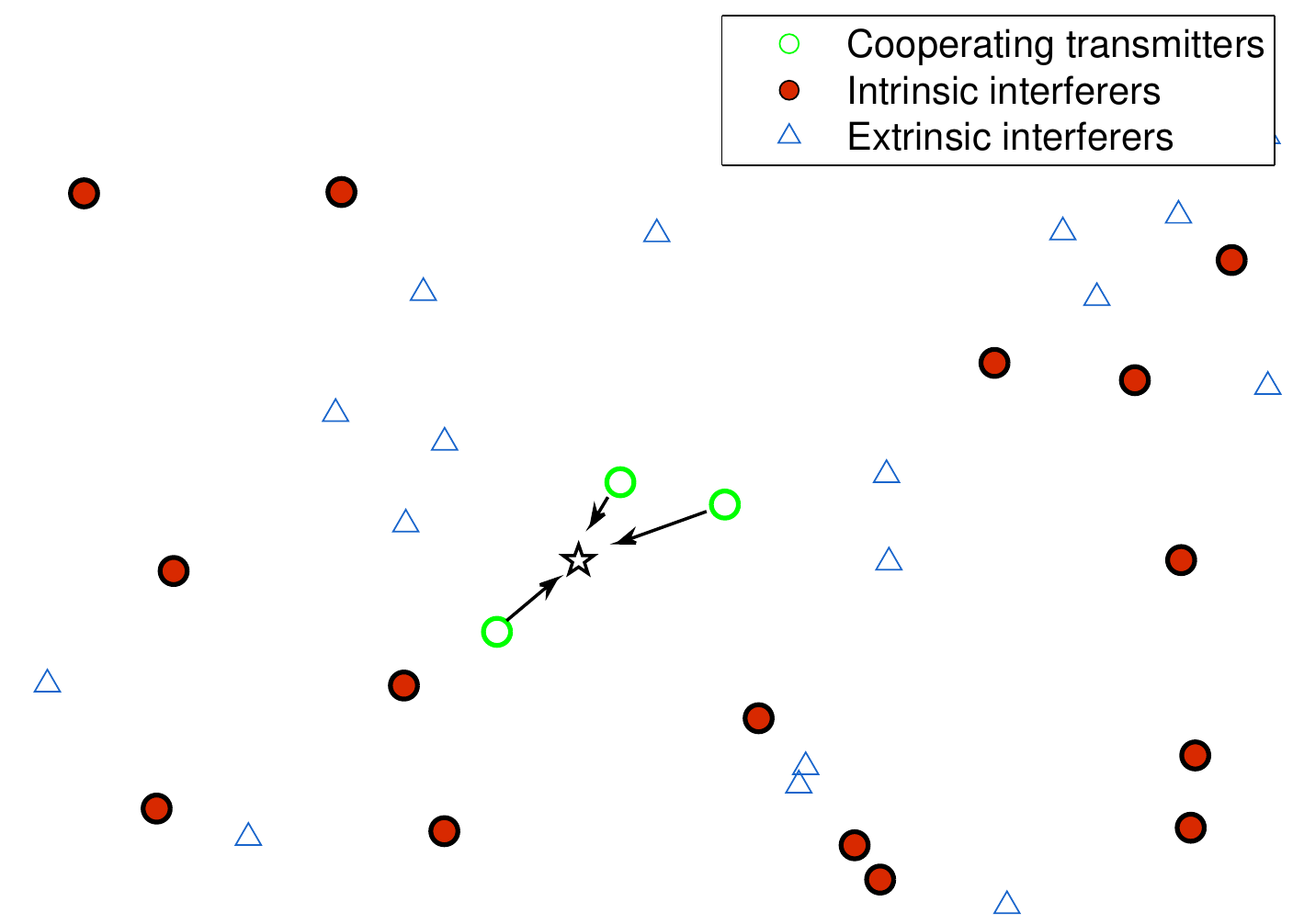}}
	\caption{A network snapshot showing a receiver \text{(\FiveStarOpen)} jointly served by a cluster of $K=3$ closest transmitters amid intrinsic and extrinsic interference. Only one receiver is shown for illustration purpose.}
	\label{fig:netmodel}
\end{figure} 
In our setup, a cluster of $K$ (self-powered) cooperating transmitters (TXs) jointly serve a desired receiver (RX) or user over the same time-frequency resource block. We assume that each user is served by the $K$ closest TXs (see Fig. \ref{fig:netmodel}). 
The TX locations are drawn from a homogeneous PPP of intensity (density) $\lambda$, which we denote as $\Phi\triangleq\{d_{i}, i\in\mathbb{N}\}$. We will refer to $\Phi$ as the TX tier. Similarly, the user locations are modeled using another PPP $\Phi_u$ of intensity $\lambda_u$, which is assumed to be independent of $\Phi$. 
When there are multiple candidate users seeking a given TX cluster, we assume that a user is selected uniformly at random for service each time-slot (see Section \ref{secClus}). We further assume that the TXs in $\Phi$ are active with a transmission probability $p_\textrm{tr}$ (Lemma \ref{lem1}, Section \ref{secEH}). For example, when a TX has no candidate user in $\Phi_u$ (i.e., void cell), it may still transmit to (wirelessly) charge other inactive users, or to serve other opportunistic users in the network (performance characterization of such users, though, is not the focus of this work). 

Consider an arbitrary user in $\Phi_u$ being jointly served by its $K$ closest TXs in $\Phi$.
Due to concurrent transmissions, it is subjected to co-channel interference from the out-of-cluster nodes. In general, we expect such a network to consist of nodes with different physical parameters (e.g., energy harvesting capability) and random locations (e.g., due to unplanned deployments). To model this heterogeneous network interference, we consider $M$ additional tiers of nodes, where the nodes in tier $m$ are located according to a homogeneous PPP $\Phi_m\triangleq\{d_{i,m},i\in\mathbb{N}\}$ of intensity $\lambda_m$, independently of other tiers. For compactness, we also introduce an alternative notation for the nodes in TX tier $\Phi$. Specifically, the subscript 0 is used while referring to the quantities of nodes in the TX tier when confusion might arise, e.g., $\Phi_0=\Phi$, $\lambda_0=\lambda$. Note that each class of nodes may differ in terms of energy harvesting rate $\rho_m$, energy buffer size $S_m$, transmit power $P_m$, and intensity $\lambda_m$. Without loss of generality, we assume the transmit power of the TX tier to be normalized
to unity. Therefore, $P_m$ also corresponds to the normalized transmit power of tier $m$, where
the normalization is done with respect to the actual transmit power of the TX tier. 
All the nodes are assumed to be equipped with single antennas. 
\subsection{Signal Model}\label{signal model}
All the nodes are assumed to employ orthogonal frequency division multiple access (OFDMA) for communication. We consider a transmission scheme where a group of $K$ cooperating TXs jointly transmit the same data to a given user over the same time-frequency resource block. Given the challenges associated with channel acquisition, none of the transmitting nodes are assumed to have any instantaneous channel knowledge. The considered joint transmission scheme is simple as it does not require joint encoding at the cooperating transmitters. To further reduce the coordination overhead, we do
not assume any tight synchronization among the in-cluster TXs. The user, however, is required to know the composite downlink channel from the in-cluster transmitters for coherent detection. The signals transmitted by the cooperating TXs superimpose non-coherently at the receiver, resulting in a received power boost. Moreover, interference seen by the user is treated as noise for the purpose of decoding. 

We now describe the channel model.
Let $H_i$ be the channel power gain for the link from a TX $i$ in $\Phi$ to the given user. We consider a rich scattering environment where all the links experience
IID narrowband Rayleigh fading such that the small-scale fading power is exponentially distributed, i.e., $H_i\,${\raise.17ex\hbox{$\scriptstyle\mathtt{\sim}$}}$\, \exp(1)$. Leveraging  Slivnyak's theorem \cite{haenggi2012stochastic}, we consider a typical user located at the origin, and characterize the performance in the presence of co-channel interference and noise.
Note that the timing offset between cooperating transmitters causes the received signal power to vary substantially across a large number of subcarriers within the coherence bandwidth (particularly when the timing offset and coherence bandwidth are assumed to be relatively large). Therefore, we can consider the average received power across these subcarriers for analysis (along the lines of \cite{tanbourgi2013tractable}).  
With such a non-coherent joint transmission scheme (see \cite[Appendix A]{tanbourgi2013tractable} for details), the
signal-to-interference-plus-noise ratio (SINR) at the user can be expressed as
\begin{align}\label{snr}
	\gamma\triangleq
	\frac{\sum\limits_{i=1}^{K}{\mathbbm{1}_{i}d_{i}^{-\eta}}H_{i}}{{I}+\sigma^2} 
\end{align}
where the Bernoulli random variable $\mathbbm{1}_{i}$ models the uncertainty due to bursty energy arrivals at the transmitter such that $\Pr\left[{\mathbbm{1}_{i}}=1\right]=p_{\text{tr},i}$ and $\Pr\left[\mathbbm{1}_{i}=0\right]=1-p_{\text{tr},i}\triangleq q_{\text{tr},i}$ for the in-cluster TXs (i.e., $1\leq i\leq K$), $\eta$ denotes the pathloss exponent, while $\sigma^2$ gives the variance of the receiver noise, which we assume to be zero-mean circularly symmetric complex Gaussian. Moreover, ${I}$ denotes the aggregate interference power observed at the receiver. 
For analytical tractability, it is assumed that the signals transmitted by the interfering nodes superimpose non-coherently at the receiver, which would typically be the case. The aggregate interference power ${I}$ can be expressed as
\allowdisplaybreaks{
\begin{align}\label{secnd}
{I}&={I}_0+\sum\limits_{m=1}^{M}{I}_m \nonumber \\
&=\underbrace{\sum\limits_{i=K+1}^{\infty}{\mathbbm{1}_{i}d_{i}^{-\eta}}H_{i}}_{\textrm{intrinsic}}+\underbrace{\sum\limits_{m=1}^{M}\sum_{d_{i,m}\in\Phi_m}{\mathbbm{1}_{i,m}}P_m d_{i,m}^{-\eta}H_{i,m}}_{\textrm{extrinsic}}
\end{align} 
}where the first term $I_0$ accounts for the in-network or intrinsic interference due to the out-of-cluster TXs in $\Phi$. 
Here, $\Pr\left[{\mathbbm{1}_{i}}=1\right]=p_{\text{tr},o}$ while $\Pr\left[{\mathbbm{1}_{i}}=0\right]=1-p_{\text{tr},o}\triangleq q_{\text{tr},o}$ for all the out-of-cluster TXs (i.e., $i>K$). 
The second term in (\ref{secnd}) models the extrinsic or out-of-network interference from the nodes belonging to the $M$ interfering tiers $\{\Phi_m\}_{m=1}^{M}$. Note that for the interfering tiers, we use a slightly modified notation by including $i,m$ in the subscript to denote a node $i$ that belongs to the interfering tier $\Phi_m$. As done for the TX tier $\Phi$, we can similarly define $\Pr\left[{\mathbbm{1}_{i,m}}=1\right]=p_{\text{tr}}^{(m)}$ for the nodes in tier $m$. The assumptions about the channel model are as explained for the TX tier $\Phi$, i.e., $H_{i,m}\,${\raise.17ex\hbox{$\scriptstyle\mathtt{\sim}$}}$\,\exp(1)$. 

\begin{table}
	\caption{Model Parameters}	
	\centering
	\begin{tabular}{| p{1.in} | p{2.0in} |}
		\hline
		\textbf{Notation} & \textbf{Description} \\ \hline
		$K$ & cluster size \\ \hline
		$\eta$ & path-loss exponent \\ \hline
		$p_{\text{ch}}$ & channel access probability \\ \hline
    	$p_{\text{tr}}\triangleq 1-q_{\text{tr}}$ & transmission probability \\ \hline
    	$\rho$ & energy harvesting rate \\ \hline
    	$S$ & energy buffer size \\ \hline
    	$\Phi_u; \lambda_u$ & PPP with intensity $\lambda_u$ modeling RX locations.\\ \hline 	
    	$\Phi; \lambda$ (or $\Phi_0, \lambda_0$) & PPP with intensity $\lambda$ (also denoted as $\lambda_0$) modeling TX locations.\\ \hline 
    	${\{p_{\text{tr,i}}\}_{i=1}^{K}}$ & transmission probabilities of $K$ in-cluster TXs in $\Phi$. \\ \hline
    	$p_{\text{tr},o}$ & transmission probability of out-of-cluster TXs in $\Phi$. \\ \hline   	
    	${\Phi_m;\lambda_m}$\newline ($1\leq m\leq M$) & PPP with intensity $\lambda_m$ modeling node locations in tier $m$. \\ \hline
    	$p_{\text{tr}}^{(m)}$\newline ($1\leq m\leq M$) & transmission probability of nodes in $\Phi_m$. \\ \hline
    	$P_m$ \newline($1\leq m\leq M$) & normalized transmit power of nodes in $\Phi_m$. \\ \hline
	\end{tabular}
	\label{tbl1}
\end{table}
\textbf{Notation.} Table \ref{tbl1} summarizes the notation introduced in this section. We adopt the following notation for the transmission probabilities of the nodes belonging to tier $\Phi$. For $i=1,\cdots,K$, we define $p_{\text{tr},i}\triangleq 1-q_{\text{tr},i}$ to be the transmission probability of the $i^{th}$ in-cluster TX belonging to $\Phi$, whereas $p_{\text{tr},o}$ gives the transmission probability of all other (i.e., out-of-cluster) TXs in $\Phi$. Similarly, for $m=1,\cdots,M$, we define $p_{\text{tr}}^{(m)}\triangleq 1-q_{\text{tr}}^{(m)}$ to be the transmission probability of the nodes belonging to the interfering tier $\Phi_m$. The above notation allows both the in-cluster and out-of-cluster nodes to have different transmission probabilities. This is in line with the considered model, where we have allowed the nodes to have possibly different energy harvesting capabilities.
For ease of exposition, we define $\Xi=\left[q_{\text{tr},1},\cdots,q_{\text{tr},K}, q_{\text{tr},o}, q_{\text{tr}}^{(1)},\cdots,q_{\text{tr}}^{(M)}\right]$, which depends on the energy harvesting parameters (i.e., energy harvesting rate and energy buffer size). We also define 
\begin{align}\label{QQ}
G=\prod\limits_{i=1}^{K}q_{\text{tr},i}.
\end{align}

For the TXs (in $\Phi$) belonging to a cluster of size $K$, we define $\omega_i=\frac{d_i}{d_K}$ such that $\{\omega_i\}_{i=1}^{K}$ denotes a set of normalized distances. This set is assumed to be arranged in ascending order, i.e., $d_1$ refers to the closest serving TX while $d_K$ refers to the TX located farthest away from the user. We also define $\Omega=\{\omega_1^\eta,\cdots,\omega_K^\eta\}$ and $\hat{\Omega}=\{\frac{\omega_1^\eta}{q_{\text{tr},1}},\cdots,\frac{\omega_K^\eta}{q_{\text{tr},K}}\}$. For generality, we allow $\Omega$ to have duplicate elements and further define the set $\{\delta_1^\eta,\cdots,\delta_\tau^\eta\}$ to consist of all the unique elements of the set $\Omega$, where $\delta_i^\eta$ occurs in $\Omega$ with multiplicity $n_i$. Note that $\tau=1,\cdots,K$, where $\tau=1$ denotes the case when $\Omega$ has identical elements, whereas $\tau=K$ when $\Omega$ has distinct elements\footnote{Note that $\Omega$ is a multiset since it may have duplicate elements. For cleaner exposition, however, we call $\Omega$ (and other multisets) a set in this paper.}.  
We further define ${K\choose K-i}_\Omega$ to be the set of all products of the elements of $\Omega$ taken $K-i$ at a time. For instance, when $K=3$, ${3\choose 1}_\Omega=\{\omega_1^\eta,\omega_2^\eta,\omega_3^\eta\}$, ${3\choose 2}_\Omega=\{\omega_1^\eta\omega_2^\eta,\omega_2^\eta\omega_3^\eta,\omega_3^\eta\omega_1^\eta\}$, and ${3\choose 3}_\Omega=\{\omega_1^\eta\omega_2^\eta\omega_3^\eta\}$. We also define a set operator ${\overset{+}{\sum}}\left[\cdot\right]$ that returns the sum of the elements of the set that it operates on. We further define
\begin{equation}\label{def4}
	\alpha_i(\Omega)={{\left(-1\right)}^{i}} {\overset{+}{\sum}}\left[{{K\choose K-i}_\Omega}\right]. 
\end{equation}
The summation in (\ref{def4}) is taken over the elements of the set ${K\choose K-i}_\Omega$. Similarly, the definition of $\alpha_i(\hat{\Omega})$ follows from (\ref{def4}) with the set $\Omega$ now replaced by $\hat{\Omega}$. For the intensity parameters, we define $\Lambda=\left[\lambda_0,\cdots,\lambda_M\right]$ (recall that $\lambda_0$ (or $\lambda$) gives the intensity of the PPP $\Phi_0$ (or $\Phi$), while $\{\lambda_m\}_{m=1}^{M}$ denote the same for the interfering tiers $\{\Phi_m\}_{m=1}^{M}$). 

\section{Stochastic Geometry Analysis}\label{secSTOCH}
In this section, we derive closed-form expressions for the complementary cumulative distribution function (CCDF) of the SINR $\gamma$ (which we call the link success probability), the cluster access probability, and the overall success probability.   
\subsection{Link Success Probability}\label{secLink}
In this subsection, we focus on the receivers which have been selected for service in a given resource. We provide closed-form expressions that characterize the CCDF of the SINR $\gamma$ at such a receiver as a function of network parameters and cluster geometry. While Theorem 1 is useful for a given cluster geometry, Theorem 2 is applicable to the general case with the absolute cluster geometry averaged out.    

\begin{theorem}\label{thm1}\normalfont
	For a cluster of size $K$, the CCDF of $\gamma$, ${\bar{F}}_\gamma\left(K,\theta\right)=\Pr\left[\gamma>\theta\right]$, can be tightly approximated as a function of the intensity parameters ($\Lambda$), noise power ($\sigma^2$), energy harvesting parameters ($\Xi$) and cluster geometry $\left(\{d_i\}_{i=1}^{K}\right)$ using
	\begin{align}\label{general}
		&{\bar{F}}_{\gamma}(K,\theta)\approx  \nonumber \\
		  & G\displaystyle \sum\limits_{u=1}^{\tau}\sum\limits_{v=1}^{n_u}\left(\sum\limits_{m=0}^{K-1}\left({\alpha_m(\hat{\Omega})}-\alpha_m(\Omega)\right)\textrm{A}_{m}(n_u,v)\right)
		\textrm{B}_{u,v}(\theta) 
	\end{align} 
	where
	\begin{align}\label{coeff}
	\textrm{A}_m(n_u,v)&=
	(-1)^{n_u-v}{\delta_u}^{-v\eta} \sum_{\sum_{i=1}^{\tau}k_i=n_u-v} 
	\binom{m}{k_u}{\delta_u}^{\eta(m-k_u)}\nonumber\\
	&\times\prod\limits_{j\neq u}^{\tau}{\binom{n_j+k_j-1}{k_j}\left({\delta_j}^\eta-{\delta_u}^\eta\right)^{-(n_j+k_j)}}.
	\end{align}
	The summation in (\ref{coeff}) is taken over all possible combinations of non-negative integer indices $k_1,\cdots,k_{\tau}$ that add up to $n_u-v$. Further,
	\begin{align}\label{B1}
			\mathrm{B}_{u,v}(\theta)=\sum\limits_{\ell=1}^{v}\binom{v}{\ell}(-1)^{\ell+1}\Delta_{u,\ell}(\theta)
    \end{align}  	
    where
	\begin{align}\label{delta}
		\Delta_{u,\ell}(\theta)={e^{-\theta \kappa \ell(d_K\delta_u)^\eta \sigma^{2}}} {e^{-\pi p_{\text{tr},o}{\lambda} d_K^2 \mathcal{F}\left({\delta_u^\eta}\theta\kappa\ell,\eta\right)}}\Psi_u\left(M\right),
	\end{align}
	\begin{align}
	\Psi_u\left(M\right)=\prod\limits_{m=1}^{M}{e^{-\pi p_{\text{tr}}^{(m)}\lambda_m {\delta_u}^2{d_K}^2{\left(\theta\kappa\ell P_m\right)}^{\frac{2}{\eta}}\Gamma\left(1+\frac{2}{\eta}\right)\Gamma\left(1-\frac{2}{\eta}\right)}},
	\end{align}
	with $\kappa={(v!)}^{-\frac{1}{v}}$ and
	\begin{align}\label{F}
		\mathcal{F}\left(t_1,t_2\right)=\frac{2\,t_1}{t_2-2} {_{2}F_1}\left(1,1-\frac{2}{t_2},2-\frac{2}{t_2},-t_1\right)
	\end{align}
	where $_{2}F_{1}(\cdot)$ is the Gauss hypergeometric function \cite{gasper2004basic}, and $\Gamma(\cdot)$ is the Gamma function.
\end{theorem} 
\begin{proof}
See Appendix A.
\end{proof}
\begin{remarks}\normalfont
Note that Theorem 1 allows the in-cluster TXs to have possibly different energy harvesting rates or buffer sizes, and is therefore useful for getting general insights about the performance when the cluster consists of heterogeneous TXs. Similarly, the multi-tier approach allows capturing the heterogeneity in out-of-cluster nodes. Furthermore, all the interfering TXs can be assumed to have the maximum harvesting rate/buffer size to get a lower bound on performance. 
\end{remarks}
\begin{corollary}\normalfont
\textit{$\theta\rightarrow0$}. In the low-outage regime, the performance is limited by the in-cluster energy harvesting parameters and the cluster size. In particular, as $\theta\rightarrow0$ in (\ref{general}), we get $\lim\limits_{\theta\to 0}{\bar{F}}_{\gamma}(K,\theta)=1-G$, where $G$ given in (\ref{QQ}) defines a limit on the performance and infact represents the exact outage probability in the asymptotic regime. In this regime, the performance is independent of the out-of-cluster parameters. This observation also holds for Theorem 2.
\end{corollary}

Theorem 1 is general in that it is applicable to any given absolute cluster geometry. Later, numerical results confirm the accuracy of the analytical expression given in Theorem 1. We now consider the case where the in-cluster distances $\{d_i\}_{i=1}^{K}$ are distinct (i.e., $\tau=K$). For this scenario, the following proposition provides a closed-form expression for the exact CCDF of the SINR $\gamma$.  
\begin{proposition}\label{prop1}\normalfont
	For a cluster of size $K$ with a distinct cluster geometry $\left(\{d_i\}_{i=1}^{K}\right)$, the CCDF of $\gamma$, ${\bar{F}}_\gamma\left(K,\theta\right)=\Pr\left[\gamma>\theta\right]$, can be expressed in terms of the intensity parameters ($\Lambda$), noise power ($\sigma^2$), and energy harvesting parameters ($\Xi$) as
	\begin{align}\label{d2}
		{}{\bar{F}}_{\gamma}(K,\theta)&= 
		G\displaystyle \sum\limits_{j=1}^{K}{\left(\frac{\sum\limits_{i=0}^{K-1}{\left({\alpha_i(\hat{\Omega})}-\alpha_i(\Omega)\right)(\omega_j^\eta)^i}}{\omega_j^\eta\left(\prod\limits_{l\neq j}^{K}\omega_l^\eta-\omega_j^\eta\right)}\right){\mathrm{C}_j(\theta)}}   
	\end{align} 
	where
	\begin{align}\label{deltaa}
		\mathrm{C}_j(\theta)={e^{-d_j^\eta \theta \sigma^{2}}} {e^{-\pi p_{\text{tr},o}{\lambda} d_K^2 \mathcal{F}\left({\omega_j^\eta}\theta,\eta\right)}}\mathrm{D}_j\left(M\right)
	\end{align}
	with
	\begin{align}
	\mathrm{D}_j\left(M\right)=\prod\limits_{m=1}^{M}{e^{-\pi p_{\text{tr}}^{(m)}\lambda_m {d_j}^2{\left(\theta P_m\right)}^{\frac{2}{\eta}}\Gamma\left(1+\frac{2}{\eta}\right)\Gamma\left(1-\frac{2}{\eta}\right)}}.
	\end{align}
\end{proposition} 
\begin{proof}
See Appendix B.
\end{proof}
\begin{corollary}\normalfont
$\{q_{\text{tr},i}\}_{i=1}^{K}=q_{\text{tr},o}\triangleq q_{\text{tr}}$. When all the TXs in $\Phi_0$ have identical energy harvesting capabilities, i.e., $q_{\text{tr},i}=q_{\text{tr},o}\triangleq q_{\text{tr}}$, the CCDF in (\ref{d2}) simplifies to
	\begin{align}\label{cor1}
		{\bar{F}}_{\gamma}(K,\theta)= 
		\lefteqn \displaystyle \sum\limits_{j=1}^{K}{\left(\frac{\sum\limits_{i=0}^{K-1}{\alpha_i(\Omega)\left({q_{\text{tr}}}^i-{q_{\text{tr}}}^K\right)(\omega_j^\eta)^i}}{\omega_j^\eta\left(\prod\limits_{l\neq j}^{K}\omega_l^\eta-\omega_j^\eta\right)}\right){\mathrm{C}_j(\theta)}} 
	\end{align}
where $\mathrm{C}_j(\theta)$ is given by (\ref{deltaa}).  
\end{corollary}
Note that Theorem 1 can be used for analyzing cooperative setups in the presence of interference and noise, for a given cluster geometry. For a network consisting of homogeneous TXs, we next provide a more general result in terms of normalized distances by unconditioning with respect to the distance $d_K$. In other words, the following result does not correspond to a particular cluster, it is rather averaged over all such clusters that share a common $\{\omega_i\}_{i=1}^{K}$.
Also note the use of superscript $\prime$ in ${\bar{F}}^{\prime}_{\gamma}(K,\theta)$ to differentiate it from the earlier notation ${\bar{F}}^{}_{\gamma}(K,\theta)$ used for Theorem 1.
\begin{theorem}\normalfont
In the interference-limited regime ($\sigma^2\rightarrow 0$), the CCDF of $\gamma$, ${\bar{F}}^{\prime}_{\gamma}(K,\theta)$, can be expressed as a function of cluster size $\left(K\right)$, intensity parameters $\left(\Lambda\right)$, and energy harvesting parameters $\left(\Xi\right)$ for a normalized cluster geometry $\left(\{\omega_i\}_{i=1}^{K}\right)$ as 
	\begin{align}\label{d3}
		{\bar{F}}^{\prime}_{\gamma}(K,\theta)&= 
		\lefteqn \displaystyle \sum\limits_{j=1}^{K}{\left(\frac{\sum\limits_{i=0}^{K-1}{\alpha_i(\Omega)\left({q_{\text{tr}}}^i-{q_{\text{tr}}}^K\right)(\omega_j^\eta)^i}}{\omega_j^\eta\left(\prod\limits_{l\neq j}^{K}\omega_l^\eta-\omega_j^\eta\right)}\right)\mathcal{V}_j\left(K,\theta\right)} 
	\end{align}
where 
\begin{align}\label{varpi}
\mathcal{V}_j\left(K,\theta\right)={\left(1+\mathcal{F}\left(\omega_j^\eta \theta,\eta\right)+ \varUpsilon_j\left(M\right)\right)}^{-K} 
\end{align}
and
\begin{align}\label{tier}
\varUpsilon_j\left(M\right)={\omega_j}^2\theta^{\frac{2}{\eta}}\Gamma\left(1+\frac{2}{\eta}\right)\Gamma\left(1-\frac{2}{\eta}\right)\sum\limits_{m=1}^{M}{\tilde{p}_{\text{tr}}^{(m)}\tilde{\lambda}_m {P_m}^{\frac{2}{\eta}}}. 
\end{align}
Note that $\mathcal{F}(\cdot,\cdot)$ follows from (\ref{F}), and we define $\tilde{p}_{\text{tr}}^{(m)}=\frac{{p}_{\text{tr}}^{(m)}}{{p}_{\text{tr}}}$ and $\tilde{\lambda}_m=\frac{{\lambda}_m}{\lambda}$.	
\end{theorem}
\begin{proof}
See Appendix C.  
\end{proof}
\begin{remarks}\normalfont
For the case with no extrinsic interference, i.e., with the $M$ interfering tiers turned off, the CCDF expression in Theorem 2 is independent of the TX intensity $\lambda$. This is because the probability of finding the closest TX around the receiver increases with $\lambda$, but so does the interference such that the two effects cancel out.
\end{remarks}
\begin{remarks}\normalfont
With the $M$ interfering tiers now turned on, the CCDF expression in Theorem 2 is no longer independent of the intensity parameters $\Lambda$. In this case, increasing the TX intensity $\lambda$ (or more generally the effective intensity $p_{\text{tr}}\lambda$) helps dilute the intensity of the interfering tiers. This is supported by (\ref{tier})
where the term inside the summation vanishes as $\lambda$ is increased. This neutralizes the harmful term $\varUpsilon_j\left(M\right)$, which captures the effect of extrinsic interference. This is in contrast to the previous case where the TX intensity $\lambda$ plays no role.
\end{remarks}
The following corollaries have been obtained assuming the $M$ interfering tiers to be turned off.
\begin{corollary}\normalfont
$p_{\text{tr}}\rightarrow 1$. It is worth noting that without energy harvesting and a random medium access protocol, i.e., as $p_{\text{tr}}\rightarrow1$ in (\ref{d3}), and further assuming the $M$ interfering tiers to be turned off, we can retrieve the expression for the CCDF of $\gamma$ in a traditionally powered cooperative network as given in \cite{lee2013base}, which Theorem 2 generalizes.
\end{corollary} 
\begin{corollary}\normalfont
\textit{$K=1$}. For the non-cooperative case, the expression in (\ref{d3}) simplifies to ${\bar{F}}_{\gamma}^{\prime}(1,\theta)=(1-q_{\text{tr}}){\left(1+\mathcal{F}(\theta,\eta)\right)}^{-1}$. 
As $p_{\text{tr}}\rightarrow 0$,  ${\bar{F}}_{\gamma}^{\prime}(1,\theta)\rightarrow 0$, which shows that the energy harvesting parameters are critical in determining outage.
Furthermore, with $q_\text{tr}=0$, we can retrieve the CCDF expression for the signal-to-interference ratio (SIR) in a traditionally powered non-cooperative network as given in \cite{andrews2011tractable}.
\end{corollary}
\begin{corollary}\normalfont
{$S\rightarrow\infty$}. As the energy buffer size $S$ goes to infinity, the transmission probability $p_\textrm{tr}$ approaches $\min\left(\rho,p_\textrm{ch}\right)$ \cite{huang2013Spatial}. Plugging $q_\textrm{tr}=1-\min\left(\rho,p_\textrm{ch}\right)$ in (\ref{d3}) yields the outage probability floor as the buffer size goes to infinity for a given outage threshold $\theta$ and cluster size $K$. For example, for $\rho<p_\textrm{ch}$, the outage probability floor is given by
\begin{align}\label{floor1}
		 &{P}^{\prime}_{\textrm{out},S\rightarrow\infty}(K,\theta)= \nonumber \\ 
		& 1-\lefteqn \displaystyle \sum\limits_{j=1}^{K}{\left(\frac{\sum\limits_{i=0}^{K-1}{\alpha_i(\Omega)\left({(1-\rho)}^i-{(1-\rho)}^K\right)(\omega_j^\eta)^i}}{\omega_j^\eta\left(\prod\limits_{l\neq j}^{K}\omega_l^\eta-\omega_j^\eta\right)}\right)\mathcal{V}_j\left(K,\theta\right)}
\end{align}     
where $\mathcal{V}_j\left(K,\theta\right)$ is as given in (\ref{varpi}). Furthermore, as $\theta\rightarrow 0$, $\lim\limits_{\theta\rightarrow 0}{P}^{\prime}_{\textrm{out},S\rightarrow\infty}(K,\theta)={\left(1-\rho\right)}^K$. This defines the minimum possible outage probability floor for a given energy harvesting rate $\rho$ and cluster size $K$ in the large energy buffer regime. 
\end{corollary}
\begin{corollary}\normalfont
{$\rho>p_\textrm{ch}$}. As the energy harvesting rate $\rho$ exceeds the channel access probability $p_\textrm{ch}$, the transmission probability $p_\textrm{tr}\rightarrow p_\textrm{ch}$ assuming a sufficiently large energy buffer, i.e., as $S\rightarrow\infty$. In this regime, the network is independent of the energy harvesting parameters and behaves like a traditionally-powered network. The outage probability expression for this regime can be obtained similar to Corollary 5. 
\end{corollary}
These results shows that our analytical framework is fairly general with the traditionally-powered cooperative and non-cooperative networks as special cases. 
Recall that this subsection characterized the link success probability assuming the receiver to be connected. To pave the way for the overall performance metric, the next subsection characterizes the cluster access probability at a typical receiver.  
\subsection{Cluster Access Probability}\label{secClus}
We define the cluster access probability $p_\textrm{clus}$ as the probability that a random user can access the desired cluster, i.e., it is selected for service in a cluster consisting of its $K$ closest transmitters. 
Since the exact analytical characterization seems challenging, we propose simple closed-form expressions to approximate $p_\textrm{clus}$. 
To illustrate the point, we consider the simple non-cooperative case $K=1$ where users connect to the closest transmitter. A typical user will be selected for service with a probability $1/n$ given that there are $n$ candidate users for the desired transmitter. Therefore, $p_\textrm{clus}=\sum\limits_{n=1}^{\infty}\frac{1}{n}\Pr\left[n\,\textrm{candidate users}\right]$. Note that for $K=1$, $\Pr\left[n\,\textrm{candidate users}\right]$ corresponds to the probability that there are $n$ users within the typical cell. Leveraging the results in \cite{Ferenc2007Voronoi,offload2013} for the probability distribution of the number of users in a cell, we get the following analytical expression 
\allowdisplaybreaks
\begin{align}\label{eq:sarab}
p_\textrm{clus}\approx\sum_{i=1}^{\infty}\frac{3.5^{3.5}}{i!}\frac{\Gamma(i+3.5)}{\Gamma(3.5)} \left(\frac{1}{\beta}\right)^{i-1} \left(3.5+\frac{1}{\beta}\right)^{-(i+3.5)}
\end{align} 
where $\beta=\frac{p_\textrm{tr}\lambda}{\lambda_u}$ (note that $p_\textrm{tr}\lambda$ is the effective transmitter density) is the density ratio. Note that this is an approximation since the area distribution of the Voronoi cell is an approximation\cite{Ferenc2007Voronoi}.
While this approach results in an analytical expression for $p_\textrm{clus}$ for $K=1$, the extension of this formulation to the cooperative case $K>1$ is rather challenging. 
\subsubsection*{Proposed Approximation}
We propose the following analytical approximation for $p_\textrm{clus}$. For a typical user, let $d_K$ be the distance to the $K$th closest transmitter. Let $x_1$ be the distance to its closest user. We can interpret $p_\textrm{clus}$ as the probability that there is no other user within a radius $c\times d_K$ of the typical user, i.e., $p_\textrm{clus}=\Pr\left[x_1>cd_K\right]$. The constant $c>0$ controls the radius of this guard zone. Consider
\begin{align}\label{eq:approx}
\mathbb{E}\left[\Pr\left[x_1>cd_K|d_K\right]\right]\overset{(a)}{=}\mathbb{E}[e^{-\lambda_u\pi c^2 d_K^2}] 
\overset{(b)}{=}{{\left[1+c^2\frac{\lambda_u}{p_\textrm{tr}\lambda}\right]}^{-K}}
\end{align}
where (a) follows by calculating the void probability in a ball of radius $cd_K$, while (b) is obtained by averaging over $d_K$, which follows a generalized Gamma distribution\cite{haenggi2005distance}. 
We propose setting $c=\sqrt{\frac{C_1(K)}{1+C_2(K)\frac{\lambda_u}{p_\textrm{tr}\lambda}}}$ in (\ref{eq:approx}), where $C_1(K)$ and $C_2(K)$ are functions of $K$. This results in the following analytical approximation
\begin{align}\label{eq:clus2}
p_\textrm{clus}(K,\beta)\approx\frac{1}{\left[1+\frac{C_1(K)}{\beta+C_2(K)}\right]^K}.
\end{align} 
Note that $p_\textrm{clus}$ is a function of the cluster size $K$ and the density ratio $\beta$. Using basic curve fitting tools, we found that the linear expressions $C_1(K)=0.06K+0.78$ and $C_2(K)=0.34K-0.49$ (for $K>1$) result in a good fit for the simulation-based $p_\textrm{clus}$ for the considered values of $K$ (see Fig. \ref{fig:clus1}). For the non-cooperative case $K=1$, $C_1(K)=0.725$ and $C_2(K)=0$ give a nice fit. Also, the proposed expression $\left[1+\frac{0.725}{\beta}\right]^{-1}$ is much simpler than the analytical approximation in (\ref{eq:sarab}) for $K=1$. The proposed approximation for $p_\textrm{clus}$ is validated in Fig. \ref{fig:clus1}. Note that the cluster access probability increases with the density ratio and decreases with the cluster size. This is because for a given cluster size, a higher density ratio suggests that the typical user requires a relatively smaller user-free guard zone, which increases the cluster access probability. 
For a given density ratio, a larger cluster size causes more receivers to compete for the same cluster, reducing the cluster access probability. 

We next introduce a performance metric that captures the combined effect of the link success probability as well as the cluster access probability.
\subsection{Overall Success Probability}\label{secSuc}
We define the overall success probability $P_\textrm{suc}(\cdot)$ as the joint probability that a user is selected for service which results in a successful packet reception, i.e., $P_\textrm{suc}(K,\theta)=\Pr\left[\gamma>\theta,\frac{x_1}{d_K}>c\right]$, where we have used the notation introduced in Section \ref{secClus}. Leveraging the results in Section \ref{secLink} and \ref{secClus}, we provide a closed-form expression for $P_\textrm{suc}(\cdot)$ in terms of the model parameters.
\begin{theorem}\normalfont
	The overall success probability $P_\textrm{suc}(K,\theta)$ can be expressed as a function of the cluster size $\left(K\right)$, the intensity parameters $\left(\Lambda,\lambda_u\right)$, and the energy harvesting parameters $\left(\Xi\right)$ for a normalized cluster geometry $\left(\{\omega_i\}_{i=1}^{K}\right)$ as 
	\begin{align}\label{eq:thm3}
	& P_\textrm{suc}(K,\theta)\approx  
	\lefteqn \displaystyle \sum\limits_{j=1}^{K}{\left(\frac{\sum\limits_{i=0}^{K-1}{\alpha_i(\Omega)\left({q_{\text{tr}}}^i-{q_{\text{tr}}}^K\right)(\omega_j^\eta)^i}}{\omega_j^\eta\left(\prod\limits_{l\neq j}^{K}\omega_l^\eta-\omega_j^\eta\right)}\right)\mathcal{Z}_j\left(K,\theta\right)} 
	\end{align}
	where 
	\begin{align}\label{thm3:varpi}
	\mathcal{Z}_j\left(K,\theta\right)={\left(1+\mathcal{F}\left(\omega_j^\eta \theta,\eta\right)+ \varUpsilon_j\left(M\right)+\frac{C_1(K)}{\beta+C_2(K)}\right)}^{-K}, 
	\end{align}
	$\mathcal{F}\left(\cdot,\cdot\right)$ and  $\varUpsilon_j\left(\cdot\right)$ are given in (\ref{F}) and (\ref{tier}), $C_1(K)$ and $C_2(K)$ follow from (\ref{eq:clus2}), and $\beta=\frac{p_\textrm{tr}\lambda}{\lambda_u}$.
\end{theorem}
\begin{proof}
	See Appendix D.  
\end{proof}
\begin{remarks}\normalfont
Note that the expression in Theorem 3 is also a function of the receiver intensity since $\beta=\frac{p_\textrm{tr}\lambda}{\lambda_r}$. When there is no extrinsic interference, the overall success probability is still a function of the ratio $\beta$. This is unlike Theorem 2, where the link success probability is independent of the intensity parameters when the $M$ interfering tiers are turned off. Also note that as the density ratio $\beta$ increases, the overall success probability in Theorem 3 approaches the link success probability in Theorem 2.
\end{remarks}
\begin{corollary}\normalfont
$\theta\rightarrow 0$. In the low outage regime, the overall success probability is limited by the energy harvesting parameters, the cluster size as well as the density ratio $\beta$. As $\theta\rightarrow 0$ in (\ref{eq:thm3}), $\lim\limits_{\theta\rightarrow0}P_\textrm{suc}(K,\theta)=p_\textrm{clus}(K,\beta)\left(1-G\right)$, where $G$ is given in (\ref{QQ}) and $p_\textrm{clus}(K,\beta)$ follows from (\ref{eq:clus2}). 
\end{corollary}
\begin{remarks}\normalfont
In the asymptotic regime $\theta\rightarrow 0$, we get fundamentally different insights on node cooperation for self-powered and traditionally-powered networks when $p_\textrm{ch}=1$.
In traditionally-powered networks, Corollary 7 implies that cooperation is in fact detrimental for the overall success probability, i.e., when $p_\textrm{tr}\rightarrow 1$, $G\rightarrow 0$ such that $\lim\limits_{\theta\rightarrow0}P_\textrm{suc}(K,\theta)=p_\textrm{clus}(K,\beta)$, which decreases as the cluster size $K$ is increased. To see this, note that $p_\textrm{clus}(K,\beta)$ in (\ref{eq:clus2}) admits a simpler approximation $p_\textrm{clus}(K,\beta)\approx 1-\frac{K}{\beta}$ when $\beta$ is sufficiently large. For a self-powered network, however, the gain due to cooperation  captured by the term $(1-G)=(1-{q_\textrm{tr}}^K)$ more than compensates for the loss due to $p_\textrm{clus}$ when $\beta$ is sufficiently high. This suggests that cooperation helps improve the overall performance in an energy harvesting network. Furthermore, due to the underlying tradeoff between the link reliability and the fraction of receivers getting served, there is an optimal cluster size that maximizes the asymptotic success probability for a given $\beta$. 
\end{remarks}

\section{Simulation Results}\label{secSIM}
In this section, we use simulations to validate the analytical results for the link success probability and the overall success probability under different network scenarios.  
We also investigate the impact of several parameters such as cluster size, energy harvesting rate and energy buffer size on the system performance to derive intuition on the system operations. 
\subsection{Link Success Probability}
\begin{figure}[thb]
	\centering
	\includegraphics[width=3.4in]{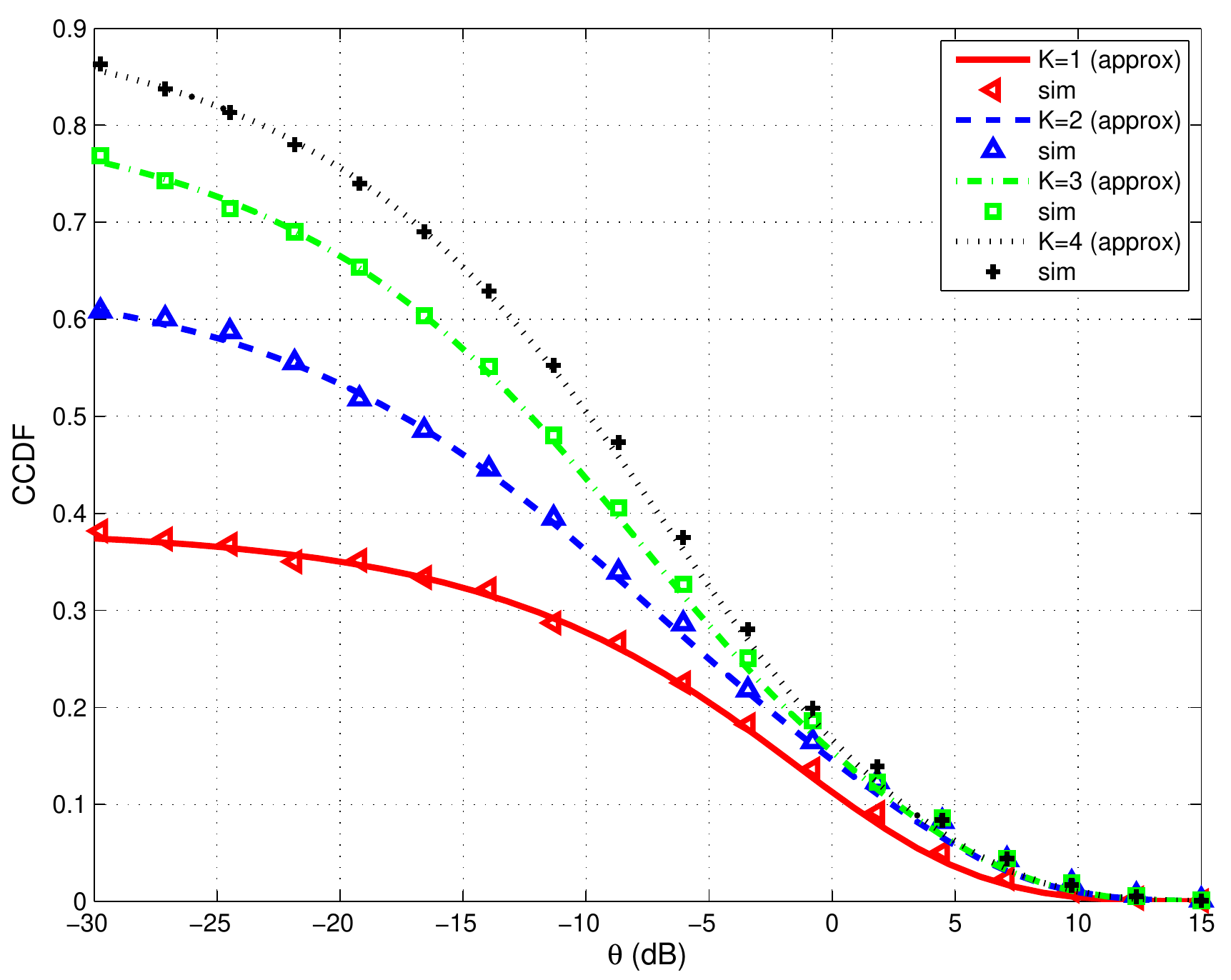}
	\caption{CCDF of $\gamma$ for various values of $K$ given $p_{\text{ch}}=0.7, \lambda=0.01, {\eta=4}$, $\sigma^2=-114$ dBm, and $\{d_i\}_{i=1}^{4}\!=\!\!\{5,10,10,10\}$. The plot is obtained for a single interfering tier $M=1$ with ${\lambda}_1=0.01$, $p_{\text{tr}}^{(1)}=0.53$ and $P_1=2$. Energy harvesting parameters for the TX tier are $\{\rho_i\}_{i=1}^{4}=\{0.4,0.45,0.5,0.55\}$, $\rho_o=0.55$ and $S=2$. Simulation (sim) results agree with the analytical approximation (approx) based on Theorem 1.}
	\label{fig:approx}
\end{figure}
	
\begin{figure}[thb]
	\centering
	\includegraphics[width=3.4in]{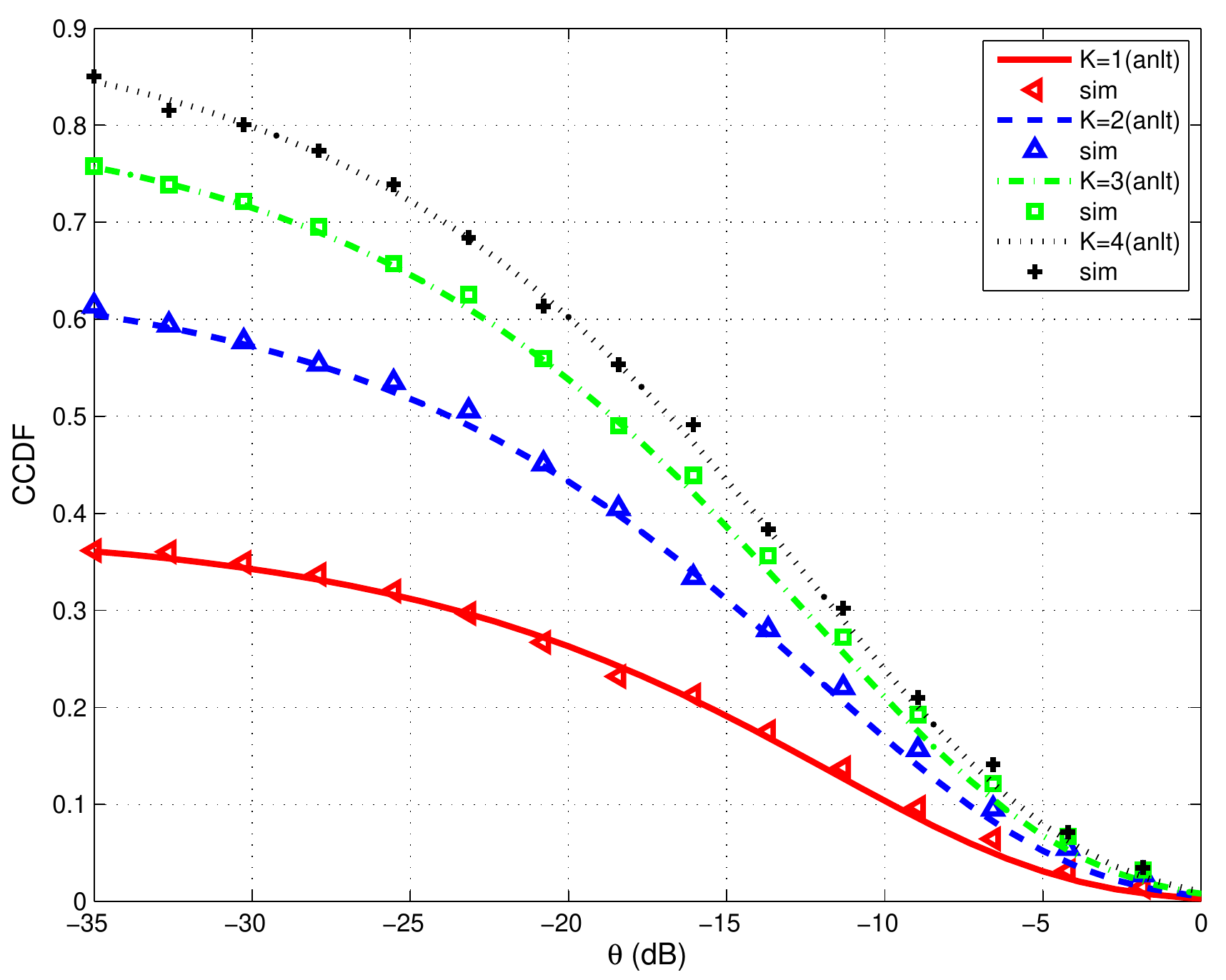}
	\caption{CCDF of $\gamma$ for various values of $K$ given a distinct cluster geometry	$\{d_i\}_{i=1}^{4}\!=\!\!\{10,12,14,16\}$. Other simulation parameters are same as in Fig. \ref{fig:approx}.
	Simulation (sim) results agree with the analytical (anlt) results based on Proposition 1.}
	\label{fig:dist}
\end{figure} 
We first consider the case with heterogeneous in-cluster TXs, and plot ${\bar{F}}_\gamma\left(K,\theta\right)$, the CCDF of $\gamma$ or the link success probability, for various values of cluster size $K$ in Fig. \ref{fig:approx}. The plot includes the curves obtained using the analytical approximation (approx) based on Theorem 1. It also includes the results obtained by Monte Carlo simulations (sim) for the given set of parameters. The analytical model is validated since there is a complete agreement between the analytical and simulation results. 
Similarly, in Fig. \ref{fig:dist}, we consider the case where the in-cluster TXs have a distinct cluster geometry. It can be observed that the simulation results match completely with the (exact) analytical (anlt) results based on Proposition 1.
  
We can draw two conclusions from Fig. \ref{fig:approx} and \ref{fig:dist}.
First, the SINR distribution at the receiver improves with $K$ due to an additional transmit diversity gain. Second, the outage performance is limited by the energy harvesting capabilities as the CCDF converges to $1-G$ in the low-outage regime ($\theta\to0$) for any given cluster. This is consistent with Corollary 1. 

Next, we consider the case where the TXs have identical energy harvesting capabilities. In Fig. 3, we plot ${\bar{F}}^{\prime}_{\gamma}(K,\theta)$, the CCDF of $\gamma$ with the absolute in-cluster distances averaged out. The plot in Fig. \ref{fig:thm2a} is obtained with the interfering tiers turned off. Note that the intensity parameter is not specified as the performance is independent of $\lambda$ for this case. It can be seen that there is a complete match between the analytical curve based on Theorem 2 and the simulated CCDF obtained via Monte Carlo simulations. A complete match between analytical and simulation results can also be observed in Fig. \ref{fig:thm2b}, which is obtained with the interfering tiers turned on.
\begin{figure*}
	\centering

	\subfloat[][]{%
		\centering
		\includegraphics[width=3.4in]{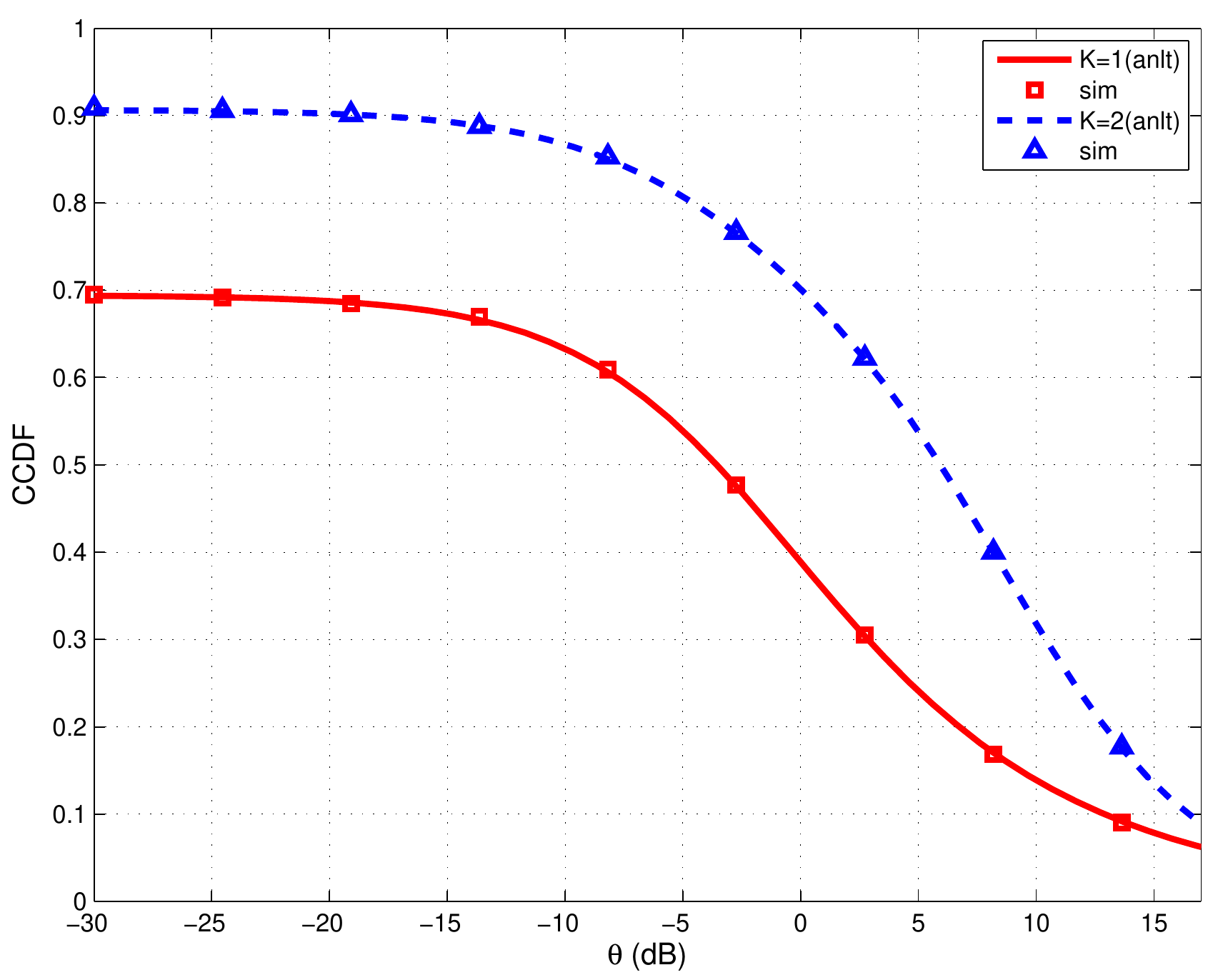}
		\label{fig:thm2a}} 
	\subfloat[][]{
		\centering
		\includegraphics[width=3.4in]{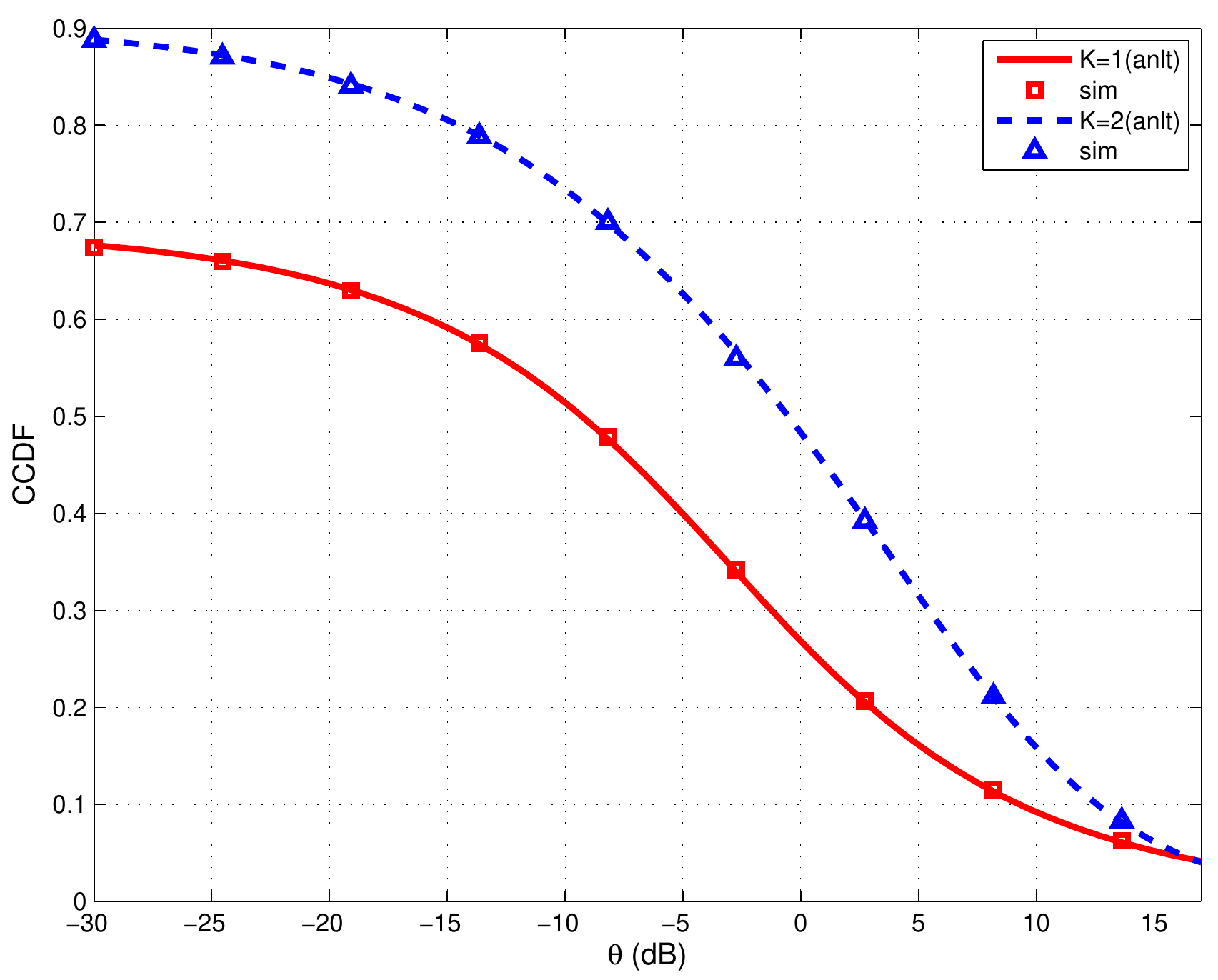}
    	\label{fig:thm2b}}	

\caption{(a) CCDF of $\gamma$ in the interference-limited regime for $K\in\{1,2\}$ with the $M$ tiers turned off (i.e., intrinsic interference only). The plot includes the analytical (anlt) curve based on Theorem 2 as well as the simulated (sim) CCDF of $\gamma$. The simulation parameters are $\omega_1=1$ for $K=1$ and $\{\omega_i\}_{i=1}^{2}=\{0.5,1\}$	for $K=2$, $p_{\text{ch}}=0.8$ and ${\eta=4}$. The energy harvesting parameters are $\rho=0.75$ and $S=2$ for all TXs. (b) For the same parameters, CCDF of $\gamma$ is plotted when both intrinsic and extrinsic interference are present. Other parameters include $M=1$, $P_1=2$, and $p_{\text{tr}}^{(1)}=0.5$. Unlike (a) which is independent of intensity, (b) is obtained for $\lambda=0.1$ and $\lambda_1=0.05$.}
\end{figure*}

As demonstrated above, the considered framework can be used to get general performance insights for a large class of self-powered wireless networks. We next study how the energy harvesting parameters limit the outage performance. 
\subsubsection*{Impact of Energy Buffer Size on Performance}
We first consider how outage probability varies as a function of energy buffer size. To get general performance insights, we use the asymptotic outage probability $P_\textrm{out}^\textrm{as}\triangleq G$, which defines an upper limit on performance given the energy harvesting parameters and cluster size. Note that similar insights can be obtained if the analysis is particularized for a given outage threshold $\theta$ using (\ref{d3}). 
In Fig. \ref{fig:c}, the asymptotic outage probability $P_\textrm{out}^\textrm{as}$ is plotted against the energy buffer size $S$ (in log scale) for various values of the cluster size $K$. We see that outage can be considerably reduced by increasing the buffer size until a limit, beyond which the curves tend to flatten out. The existence of this outage probability floor follows from Corollary 5, and the floor value is specified by (\ref{floor1}). It appears that appreciable performance gains can be extracted with a relatively small buffer size. Moreover, the benefits of having a high-capacity energy buffer tend to increase with the cluster size as depicted by the increasing steepness of the slopes (when $S$ is small) as $K$ is increased. 
This interplay between cluster and buffer size also suggests that the extent of cooperation could influence the design of energy harvesting devices, even though the energy harvesting process is assumed to be independent across the cooperating TXs. 
In addition, we observe that the outage is reduced by roughly an order of magnitude with every addition in the cluster size. 

\begin{figure*}[thb]
	\centering
		\subfloat[][]{%
			\centering
			\includegraphics[width=3.4in]{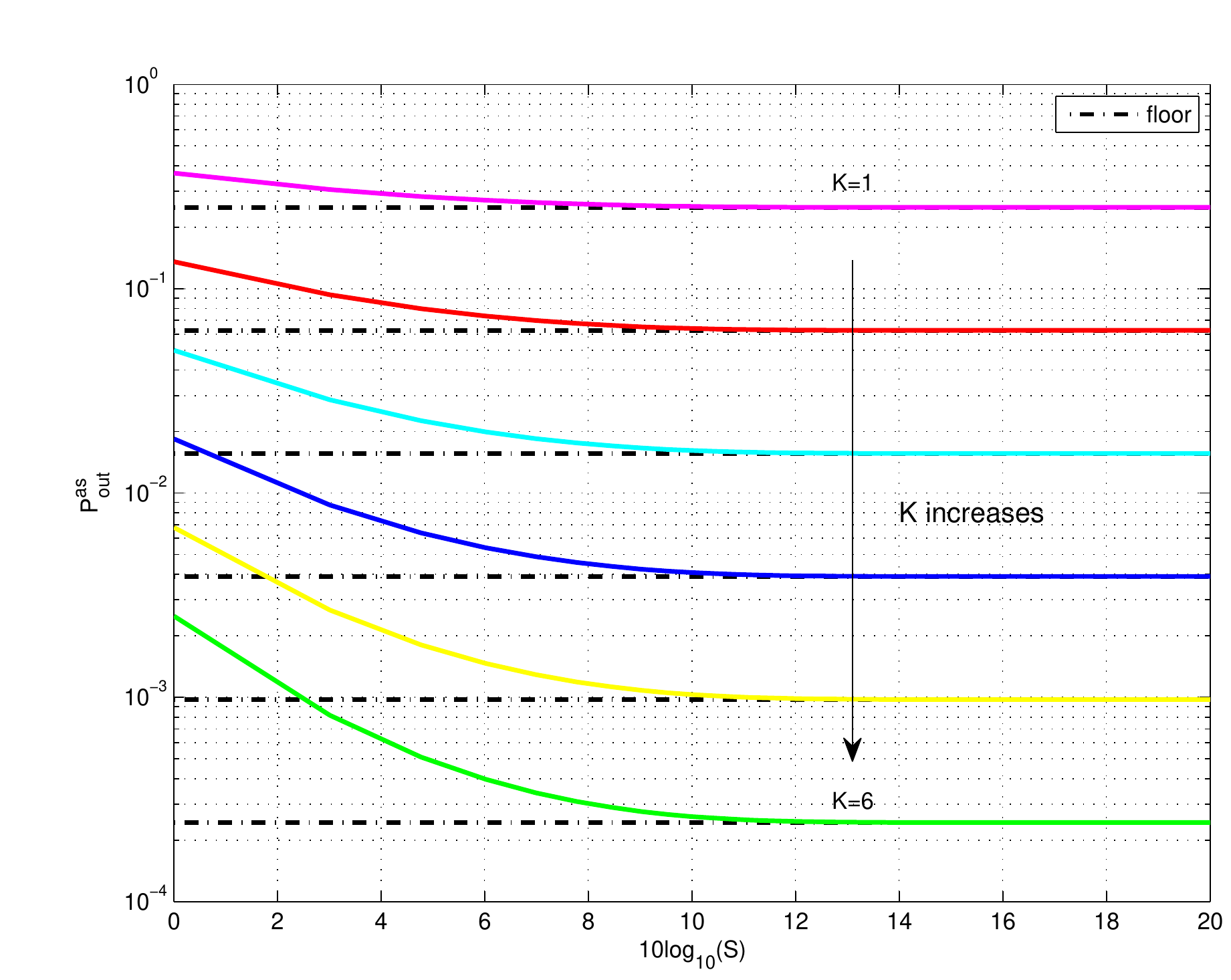}
			\label{fig:c}} 
		\subfloat[][]{
			\centering
			\includegraphics[width=3.4in]{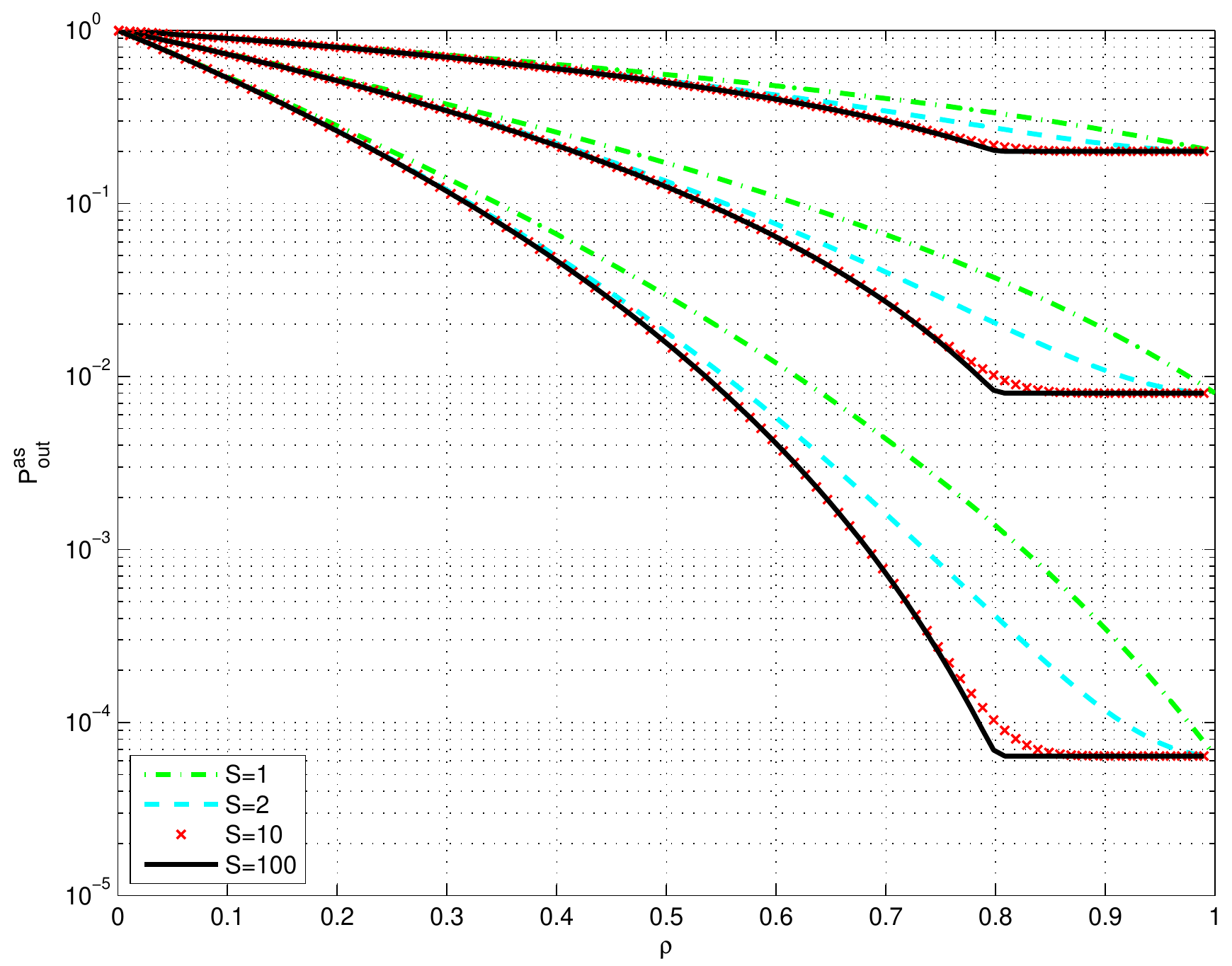}
	    	\label{fig:d}}	
\caption{(a) Impact of energy buffer size $S$ on asymptotic outage probability $P_\textrm{out}^\textrm{as}$ for various values of $K$ at fixed $\rho=0.75$ and $p_\text{ch}=0.8$. The outage probability floor is based on Corollary 5. The utility of having a larger energy buffer increases with the cluster size. (b) Impact of energy harvesting rate $\rho$ on asymptotic outage probability $P_\textrm{out}^\textrm{as}$ for various values of energy buffer size $S$ at fixed $p_\text{ch}=0.8$. The curves are plotted for cluster size $K\in\{1,3,6\}$. The outage performance becomes independent of the energy harvesting rate as the latter exceeds the channel access probability for sufficiently large energy buffers.}	
\end{figure*}
\subsubsection*{Impact of Energy Harvesting Rate on Performance}
In Fig. \ref{fig:d}, the asymptotic outage probability $P_\textrm{out}^\textrm{as}$ is plotted against the energy harvesting rate $\rho$ for various values of energy buffer size $S$. We observe that outage reduces with the increase in energy harvesting rate at the transmitters. Moreover, using a larger energy buffer brings about further reduction in outage due to enhanced energy availability at the transmitters. Furthermore, the gains from using a larger buffer size are more evident at relatively high energy harvesting rates. Fig. \ref{fig:d} also corroborates the previous observation (cf. Fig. \ref{fig:c}) that substantial performance can be extracted by using a relatively small buffer size. For example, $S=10$ suffices for this setup.
In addition, if the energy harvesting rate $\rho$ exceeds the channel access probability $p_\text{ch}$, and the buffer size is allowed to increase, the outage performance limit becomes independent of the energy harvesting rate $\rho$. For example, this behavior is evident in Fig. \ref{fig:d} for $S=100$. This is because under these conditions, the energy harvesting system tends to behave like a traditionally powered system. This is consistent with Corollary 6.   
%

Note that the previous results have been obtained for the link success probability. Another useful performance metric is the overall success probability, which is discussed next. 
\subsection{Overall Success Probability}
We first validate the analytical approximation for the cluster access probability $p_\textrm{clus}(K,\beta)$ proposed in (\ref{eq:clus2}). In Fig. \ref{fig:clus1}, there is a nice agreement between analytical and Monte Carlo simulation-based results. Moreover, in line with Corollary 7, $p_\textrm{clus}(K,\beta)$ decreases with the cluster size $K$ and increases with the density ratio $\beta$. 
\begin{figure}[thb]
	\centering
	\includegraphics[width=3.4 in]{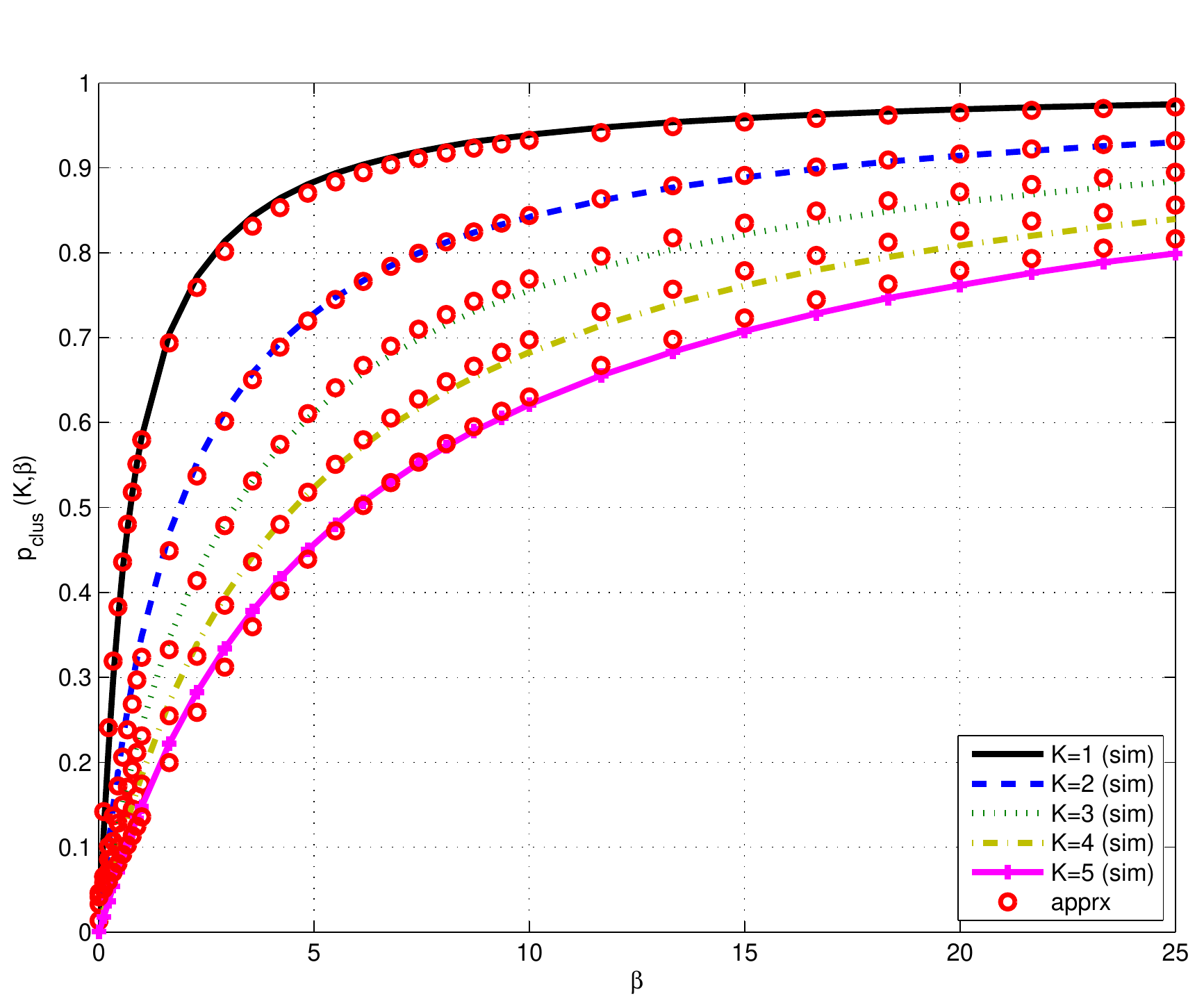}
	\caption{Cluster access probability $p_\textrm{clus}(K,\beta)$ as a function of density ratio $\beta$ for different values of the cluster size $K$. The results based on analytical approximation (apprx) in (\ref{eq:clus2}) closely match the simulation (sim) results. $p_\textrm{clus}$ also gives the overall success probability for a traditionally-powered network in the low outage regime ($\theta\rightarrow 0$). The results are in line with Remark 5.}
	\label{fig:clus1}
\end{figure} 

The overall success probability is plotted in Fig. \ref{fig:clus2}. The plots shows that cooperation is generally beneficial for the overall success probability in self-powered networks (unlike the traditional case as discussed in Remark 5). 
Moreover, it also captures the underlying tension between two competing metrics, the link performance and the fraction of receivers getting served.
As explained in Remark 5, this leads to an optimal cluster size that maximizes the overall performance. 
We further observe that the optimal cluster size increases with the density ratio due to an underlying increase in the cluster access probability.  
Though the plot is obtained for the asymptotic case $\theta\rightarrow 0$, similar trends can be observed when the analysis is particularized for a given value of $\theta$.  
\begin{figure}[thb]
	\centering
	\includegraphics[width=3.4 in]{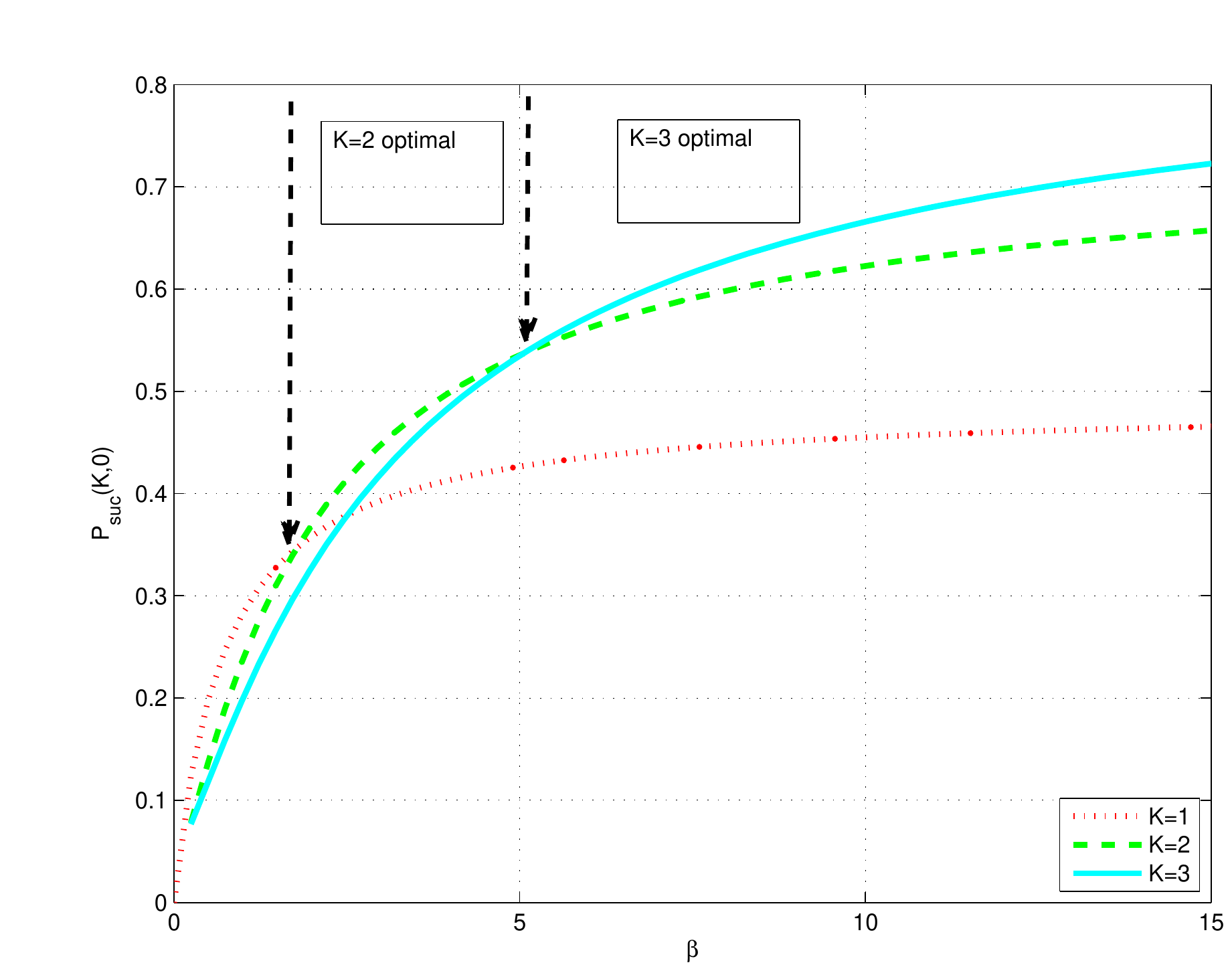}
	\caption{Overall success probability $P_\textrm{suc}(K,\theta)$ in the asymptotic regime $(\theta\rightarrow 0)$ as a function of the density ratio $\beta$ for different values of the cluster size $K$. Cooperation improves performance in a self-powered network. There is an optimal cluster size $K$ for a given value of $\beta$. The energy harvesting parameters are fixed to $S=2$, and $\rho=0.5$, and $p_\textrm{ch}=0.8$.}
	\label{fig:clus2}
\end{figure} 
\section{Conclusions}\label{secCONCLUSION}
We have considered a large-scale cooperative wireless network where clusters of self-powered transmitters jointly serve a desired receiver.
Using stochastic geometry, we have provided a tractable analytical framework to characterize the link and network-level performance at the receiver amid heterogeneous network interference. The analysis leads to several useful insights on system operation. First, the overall success probability might degrade when the cluster size is too large, unlike the link success probability which improves with the cluster size. This is due to the underlying tradeoff between the link quality and the fraction of receivers getting served.
Moreover, the resulting optimal cluster size increases with the density ratio due to the underlying improvement in the cluster access probability. 
Second, we get fundamentally different insights on node cooperation  in self-powered and traditionally-powered networks. In particular, in the asymptotic regime where the link outage threshold is small, it is optimal not to cooperate in a traditionally-powered network. In a self-powered network, however, cooperation could be beneficial since it helps overcome the performance loss due to uncertain energy availability at the transmitter. Third, the overall performance improves with the buffer size and the energy harvesting rate. Furthermore, most performance gains can be extracted using a relatively small buffer size, with the improvement becoming more pronounced for large clusters in sufficiently dense networks.  

 

\appendices
\allowdisplaybreaks
\section*{Appendix A: Derivation of Theorem 1}
We first state a lemma that will be used in the following derivation. 
\label{App:Lemma_1}
\renewcommand{\theequation}{A.\arabic{equation}}
\setcounter{equation}{0}
\begin{lemma}\normalfont
For a non-negative integer $n$, and a positive real number $x$, the regularized upper incomplete Gamma function $\mathcal{Q}(n,x)$ can be upper bounded as $\mathcal{Q}(n,x)\leq 1-(1-e^{-c x})^n$,
where $c={(n!)}^{-\frac{1}{n}}$.
\end{lemma}
\begin{proof}
See \cite{alzer1997some}.
\end{proof}

Using (\ref{snr}), we write ${\bar{F}}_\gamma\left(K,\theta\right)=\Pr\left[\gamma>\theta\right]=\mathbb{E}\left[\Pr\left[S_K >\theta d_K^\eta\left(I+{\sigma}^{2}\right)\right]\right]$, where $S_K=\sum\limits_{i=1}^{K}{\mathbbm{1}_{i}}\hat{H}_{i}$ and $\hat{H}_{i}={H}_{i}\omega_{i}^{-\eta}$. 
To proceed further, we first find the CCDF of $S_K$, where $S_K$ is a sum of $K$ independent random variables. Note that $\hat{H}_{i}$ is exponentially distributed with mean $\omega_{i}^{-\eta}$, whereas the indicator ${\mathbbm{1}_{i}}$ follows a Bernoulli distribution with mean $p_{\text{tr},i}$, independently of $\hat{H}_{i}$. The CCDF of $S_K$ can be expressed as ($x \geq 0$)
	\begin{align}\label{ccdf1}
	&{\bar{F}}_{S_k}\left(x\right)= \nonumber\\
	& G\displaystyle \sum\limits_{u=1}^{\tau}\sum\limits_{v=1}^{n_u}\left(\sum\limits_{m=0}^{K-1}\left({\alpha_m(\hat{\Omega})}-\alpha_m(\Omega)\right)\textrm{A}_{m}(n_u,v)\right){\mathcal{Q}(v,\delta_u^\eta x)}		 
	\end{align} 
where $\textrm{A}_{m}(n_u,v)$ is given by (\ref{coeff}), while $\mathcal{Q}(a,b)=\frac{1}{\Gamma(a)}\int\limits_{b}^{\infty}t^{a-1}e^{-t}\text{d}t$ denotes the regularized upper incomplete Gamma function.
The expression in (\ref{ccdf1}) can be obtained by finding the characteristic function of $S_K$, applying partial fraction expansion and then taking the inverse transform \cite{talhaLet}. Conditioning on the aggregate interference power $I$, we can write ${\bar{F}}_{\gamma|I}\left(K,\theta\right)={\bar{F}}_{S_K}\left(\theta d_K^\eta(I+\sigma^2)\right)$. 
Using (\ref{ccdf1}), and by unconditioning with respect to $I$, ${\bar{F}}_{\gamma}\left(K,\theta\right)$ can be expressed as (for $\theta\geq 0$)
\allowdisplaybreaks[0]{
\begin{align}\label{d22}
	{\bar{F}}_{\gamma}(K,\theta)=  
	G\displaystyle &
	 \sum\limits_{u=1}^{\tau}\sum\limits_{v=1}^{n_u}\left(\sum\limits_{m=0}^{K-1}\left({\alpha_m(\hat{\Omega})}-\alpha_m(\Omega)\right)\textrm{A}_{m}(n_u,v)\right) \nonumber\\
	&\times
	 \mathbb{E}\left[\mathcal{Q}\left(v,\delta_u^{\eta} d_K^\eta \theta\left(I+\sigma^{2}\right) \right)\right]
\end{align} 
}where the expectation in (\ref{d22}) is over the aggregate interference power $I$, i.e., over both fading and interferer locations. 
	\begin{figure*}[!t]
		\normalsize
		\renewcommand{\theequation}{C.\arabic{equation}}
		\setcounter{mytempeqncnt}{\value{equation}}
		\setcounter{equation}{1}
		\begin{align}\label{intg}
			\mathbb{E}\left[\mathrm{C}_j(\theta)\right]&=\int\limits_{r>0}^{}e^{-\pi p_{\text{tr}}{\lambda} r^2 \mathcal{F}\left({\omega_j^\eta}\theta,\eta\right)} 
			\prod\limits_{m=1}^{M}{e^{-\pi p_{\text{tr}}^{(m)}\lambda_m {\omega_j}^{2} {\left(P_m \theta\right)}^{\frac{2}{\eta}}\Gamma\left(1+\frac{2}{\eta}\right)\Gamma\left(1-\frac{2}{\eta}\right)}}
			\frac{2{(p_{\text{tr}}\lambda\pi r^2)}^Ke^{-p_{\text{tr}}\lambda\pi r^2}}{r\Gamma(K)} \text{d}r \nonumber \\
			&\overset{}{=}\int\limits_{0}^{\infty} \frac{e^{-\upsilon}\upsilon^{K-1}}{\Gamma(K)\left(1+\mathcal{F}\left(\omega_j^\eta \theta,\eta\right)+{\omega_j}^{2} {\theta}^{\frac{2}{\eta}}\Gamma\left(1+\frac{2}{\eta}\right)\Gamma\left(1-\frac{2}{\eta}\right)\sum\limits_{m=1}^{M}{\tilde{p}_{\text{tr}}^{(m)}\tilde{\lambda}_m {P_m}^{\frac{2}{\eta}}}\right)^{K}} \text{d}\upsilon \nonumber \\
			&=\frac{1}{\left(1+\mathcal{F}\left(\omega_j^\eta \theta,\eta\right)+{\omega_j}^{2}{\theta}^{\frac{2}{\eta}}\Gamma\left(1+\frac{2}{\eta}\right)\Gamma\left(1-\frac{2}{\eta}\right)\sum\limits_{m=1}^{M}{\tilde{p}_{\text{tr}}^{(m)}\tilde{\lambda}_m {P_m}^{\frac{2}{\eta}}}\right)^{K}}   
		\end{align}
		\begin{align}\label{dummy}
			\upsilon=\pi r^2\left(p_{\text{tr}}{\lambda}\left(1+\mathcal{F}\left({\omega_j^\eta}\theta,\eta\right)\right)+{\omega_j}^{2}{\theta}^{\frac{2}{\eta}}\Gamma\left(1+\frac{2}{\eta}\right)\Gamma\left(1-\frac{2}{\eta}\right)\sum\limits_{m=1}^{M}{\tilde{p}_{\text{tr}}^{(m)}\tilde{\lambda}_m{P_m}^{\frac{2}{\eta}}}\right)
		\end{align}
		\setcounter{equation}{\value{mytempeqncnt}}
		\hrulefill
		\vspace*{4pt}
	\end{figure*} 
A series expansion of the incomplete Gamma function gives the following alternative form  
\begin{align}\label{impos}
\mathbb{E}&\left[\mathcal{Q}\left(v,\delta_u^{\eta}d_K^\eta \theta\left(I+\sigma^{2}\right) \right)\right]= \nonumber\\
&
\sum\limits_{i=0}^{v-1}
e^{-\delta_u^\eta d_K^\eta \theta\sigma^{2}}
\mathbb{E}\left[e^{-\delta_u^\eta d_K^\eta \theta I}
\frac{\left(\delta_u^\eta d_K^\eta \theta\left(I+\sigma^{2}\right)\right)^i}{i!}
\right]. 
\end{align}
To avoid directly dealing with the expectation in (\ref{impos}), which seems rather unwieldy, we leverage the upper bound for $\mathcal{Q}\left(\cdot,\cdot\right)$ given in Lemma 2.
\begin{align}
\mathbb{E}&\left[\mathcal{Q}\left(v,\delta_u^{\eta}d_K^\eta \theta\left(I+\sigma^{2}\right) \right)\right]\leq \mathbb{E}\left[1-\left(1-e^{-\kappa\delta_u^{\eta}d_K^\eta \theta\left(I+\sigma^{2}\right)}\right)^v\right] \nonumber\\
&
=\mathbb{E}\left[\sum\limits_{\ell=1}^{v}\binom{v}{\ell}(-1)^{\ell+1}e^{-\kappa\ell\delta_u^{\eta}d_K^\eta \theta\left(I+\sigma^{2}\right)} \right] \nonumber\\
&
=\sum\limits_{\ell=1}^{v}\binom{v}{\ell}(-1)^{\ell+1}e^{-\kappa\ell\delta_u^{\eta}d_K^\eta \theta\sigma^{2}}
\mathbb{E}\left[e^{-\kappa\ell\delta_u^{\eta}d_K^\eta \theta I}\right]
\end{align}
where $\kappa={\left(v!\right)}^{-\frac{1}{v}}$, and the last equation follows by applying Binomial theorem. The next step is to evaluate the expectation $\mathbb{E}\left[e^{-\kappa\ell\delta_u^{\eta}d_K^\eta \theta I}\right]$. Since the PPPs $\{\Phi_m\}_{m=0}^{M}$ are assumed to be independent, it follows that
\allowdisplaybreaks
	\begin{align}\label{d44}
	\mathbb{E}\left[e^{-\kappa\ell\delta_u^{\eta}d_K^\eta \theta I}\right]=\mathbb{E}\left[e^{\kappa\ell\delta_u^{\eta}d_K^\eta \theta I_0}\right]\prod\limits_{m=1}^{M}\mathbb{E}\left[e^{-\kappa\ell\delta_u^{\eta}d_K^\eta \theta I_m}\right]
	\end{align}
where the first term (with $m=0$) corresponds to intrinsic interference, whereas the remaining terms ($m\geq1$) correspond to extrinsic interference. The expectation in (\ref{d44}) can be evaluated using the Laplace transform of $I_m$, which we denote by $\mathcal{L}_{I_m}(s)=\mathbb{E}\left[e^{-sI_m}\right]$.
\allowdisplaybreaks{\begin{align}\label{d4}
	\mathcal{L}_{I_m}(s)&\overset{}{=}\mathbb{E}\left[e^{-s \left( \sum\limits_{d_{i}\in \Phi_m \setminus \mathcal{B}(g_m)}{P_m \mathbbm{1}_{i} H_{i}d_{i}^{-\eta}}\right)}\right] \nonumber \\
	&
	\overset{(a)}{=}\mathbb{E}\left[\prod\limits_{d_{i}\in \hat{\Phi}_m \setminus \mathcal{B}(g_m)}\text{E}\left[e^{-s P_m H_{i}d_{i}^{-\eta}}\right]\right] \nonumber \\
	&\overset{(b)}{=}\mathbb{E}\left[\prod\limits_{d_{i}\in \hat{\Phi}_m \setminus \mathcal{B}(g_m)}^{}\frac{1}{1+s P_m d_{i}^{-\eta}}\right] \nonumber \\
	&
	\overset{}{=}\exp\left(-2\pi\hat{\lambda}_m \int_{g_m}^{\infty}\frac{x}{1+{s}^{-1}{P_m}^{-1}x^\eta}\text{d}x \right)
	\end{align}
	}where $\mathcal{B}(g_m)$ denotes a disc of radius $g_m$ centered at origin, and is used to model an interference-free guard zone around the user with respect to tier $m$. The inner expectation in $(a)$ is over fading power while the outer expectation is over the PPP $\Phi_m$ of intensity $\lambda_m$ outside $\mathcal{B}(g_m)$. Next, we exploit the property of independent thinning of a PPP to deal with the transmission indicator and consider a (thinned) PPP $\hat{\Phi}_m$ with effective density $\hat{\lambda}_m= p_{\text{tr}}^{(m)} \lambda_m$ for $1\leq m\leq M$ and $\hat{\lambda}_m= p_{\text{tr,o}}^{} \lambda_m$ for $m=0$. As the fading is IID across links and from further conditioning over the location, we obtain $(b)$. The last equation follows by invoking the probability generating functional (PGFL) \cite{haenggi2012stochastic} of the PPP and by further algebraic manipulations. With some additional algebraic steps, (\ref{d4}) can be expressed in terms of a hypergeometric function, which with $s=\kappa\ell \delta_{u}^{\eta} d_K^{\eta} \theta$ gives
\allowdisplaybreaks{
		\begin{align}\label{L}
		\mathcal{L}_{I_m}(s)|_{s=\kappa\ell \delta_{u}^{\eta} d_K^{\eta} \theta}=\exp\left(-\pi\hat{\lambda}_m g_m^2
		\mathcal{F}\left(\frac{\delta_u^\eta d_K^\eta}{g_m^\eta}P_m\kappa\ell\theta,\eta\right)
		\right)
		\end{align} 
}where $\mathcal{F}(\cdot,\cdot)$ is given by (\ref{F}). 	
To compute the expectation of the term in (\ref{d44}) arising due to out-of-cluster TXs in $\Phi_0$, set $g_0=d_K$. This is because the cluster is assumed to consist of the $K$ closest nodes and interference is due to the nodes located outside this protection zone. For the interfering tiers $\{\Phi_m\}_{m=1}^{M}$, however, no such protection zone is assumed. Without an interferer-free protection zone (i.e., $g_m\to0$), $\mathcal{L}_{I_m}(s)$ further simplifies to
\begin{align}\label{Lap}
\mathcal{L}_{I_m}&(s)|_{s=\kappa\ell\delta_u^\eta d_K^\eta\theta}= \nonumber \\
&
\exp\left(-\pi\hat{\lambda}_m {\delta_u}^2 {d_K}^2\Gamma(1+2/\eta)\Gamma(1-2/\eta)(P_m\theta\kappa\ell)^{2/\eta}\right).
\end{align}
Evaluating the expectation in (\ref{d22}) using (\ref{L}), (\ref{Lap}), and further substituting $P_0=1$, $\hat{\lambda}_0=p_{\text{tr},o}\lambda_0$, and $\hat{\lambda}_m=p_{\text{tr}}^{(m)}\lambda_m$, we obtain the result in Theorem 1. \qed

\section*{Appendix B: Derivation of Proposition 1}
\label{App:Lemma_2}
\renewcommand{\theequation}{B.\arabic{equation}}
\setcounter{equation}{0}
Similar to Appendix A, we can express ${\bar{F}}_\gamma\left(K,\theta\right)=\mathbb{E}\left[\Pr\left[S_K >\theta d_K^\eta\left(I+{\sigma}^{2}\right)\right]\right]$, where $S_K=\sum\limits_{i=1}^{K}{\mathbbm{1}_{i}}\hat{H}_{i}$ and $\hat{H}_{i}={H}_{i}\omega_{i}^{-\eta}$. 
Recall that Proposition 1 is specialized to the case where the set $\Omega$ consists of distinct elements (i.e., $\tau=K$). Given a distinct cluster geometry, the CCDF of $S_K$ given in (\ref{ccdf1}) can be further simplified to the following form (for $x \geq 0$)
	\begin{align}\label{ccdf11}
	{\bar{F}}_{S_k}\left(x\right)&= G\displaystyle \sum\limits_{j=1}^{K}{\left(\frac{\sum\limits_{i=0}^{K-1}{\left({\alpha_i(\hat{\Omega})}-\alpha_i(\Omega)\right)(\omega_j^\eta)^i}}{\omega_j^\eta\left(\prod\limits_{l\neq j}^{K}\omega_l^\eta-\omega_j^\eta\right)}\right){e^{-\omega_j^\eta x}}} 
	\end{align} 
where (\ref{ccdf11}) follows by plugging $\tau=K$ and $n_u=1$ (for $u=1,\cdots,\tau$) in (\ref{ccdf1}).
Conditioning on the aggregate interference power $I$, we can write ${\bar{F}}_{\gamma|I}\left(K,\theta\right)={\bar{F}}_{S_K}\left(\theta d_K^\eta(I+\sigma^2)\right)$. 
Using (\ref{ccdf11}), and taking expectation with respect to $I$, we can express ${\bar{F}}_{\gamma}\left(K,\theta\right)$ (for $\theta\geq 0$) as
	\allowdisplaybreaks[0]{
	\begin{align}\label{d222}
	{\bar{F}}_{\gamma}&(K,\theta)= \nonumber \\ 
	& 
	G\displaystyle \sum\limits_{j=1}^{K}{\left(\frac{\sum\limits_{i=0}^{K-1}{\left({\alpha_i(\hat{\Omega})}-\alpha_i(\Omega)\right)(\omega_j^\eta)^i}}{\omega_j^\eta\left(\prod\limits_{l\neq j}^{K}\omega_l^\eta-\omega_j^\eta\right)}\right) \mathbb{E}\left[{e^{-\omega_j^\eta d_K^\eta \theta\left(I+\sigma^{2}\right)}} \right]}   
	\end{align} 
}where the expectation in (\ref{d222}) is over the aggregate interference power $I$. Unlike Appendix A where an approximation was used, the expectation in (\ref{d222}) can be directly evaluated using (\ref{d4}).
Since the PPPs are assumed to be independent, it follows that
\begin{align}\label{d33}
	\mathbb{E}\left[e^{-{d_j}^{\eta} \theta I}\right]=\mathbb{E}\left[e^{-d_j^\eta \theta I_0}\right]\prod\limits_{m=1}^{M}\mathbb{E}\left[e^{-d_j^\eta \theta I_m}\right].
\end{align}

As the rest of the derivation follows directly from Appendix A, some steps are omitted. Evaluating (\ref{d4}) at $s=d_j^\eta \theta$ yields
\begin{align}\label{LL}
		\mathcal{L}_{I_m}(s)|_{s=d_j^\eta\theta}=\exp\left(-\pi\hat{\lambda}_m g_m^2
		\mathcal{F}\left(\frac{d_j^\eta}{g_m^\eta}P_m\theta,\eta\right)
		\right).
\end{align} 
To compute the expectation of the term in (\ref{d33}) arising due to intrinsic interferers in $\Phi_0$, set $g_0=d_K$ in (\ref{LL}). Similarly,
for the interfering tiers $\{\Phi_m\}_{m=1}^{M}$, the expectation in (\ref{d33}) is given by 
\begin{align}\label{Lapp}
\mathcal{L}_{I_m}&(s)|_{s=d_j^\eta\theta}= \nonumber \\
&
\exp\left(-\pi\hat{\lambda}_m {d_j}^2\Gamma(1+2/\eta)\Gamma(1-2/\eta)(P_m\theta)^{2/\eta}\right).
\end{align}
Evaluating the expectation in (\ref{d222}) using (\ref{LL}), (\ref{Lapp}), and further substituting $P_0=1$, $\hat{\lambda}_0=p_{\text{tr},o}\lambda_0$, $\hat{\lambda}_m=p_{\text{tr}}^{(m)}\lambda_m$ and $d_j=\omega_j d_K$, yields the result in Proposition 1. \qed
\section*{Appendix C: Derivation of Theorem 2}
\label{App:Lemma_3}
\renewcommand{\theequation}{C.\arabic{equation}}
\setcounter{equation}{0}
	We begin the proof along the lines of \cite{lee2013base} by leveraging a known result on the PPP distance distribution. As shown in \cite{haenggi2005distance}, the distance $d_K$, between a typical user and its $K{\text{th}}$ closest TX, follows a generalized Gamma distribution, i.e., 
	\begin{align}\label{eq:c1}
	f_{d_K}(r)={\frac{2}{r\Gamma(K)}{\left(p_{\text{tr}}{\lambda}\pi r^2\right)}^Ke^{-p_{\text{tr}}{\lambda}\pi r^2}}.
	\end{align}
Plugging $\sigma^2=0$ in (\ref{deltaa}), and taking expectation with respect to $d_K$, we arrive at the expression in (\ref{intg}) (given at the top of the page)
where the last equation is obtained by using a dummy variable (given in (\ref{dummy}))
for integration, and using the definition of the Gamma function $\Gamma(K)=\int\limits_{0}^{\infty}e^{-x}x^{K-1}\text{d}x$. Unconditioning (\ref{cor1}) with respect to $d_K$, and using (\ref{intg}), we recover Theorem 2. \qed
\addtocounter{equation}{2}
\section*{Appendix D: Derivation of Theorem 3}
\label{App:AppendixD}
\renewcommand{\theequation}{C.\arabic{equation}}
\setcounter{equation}{0}
Leveraging the notation from Section \ref{secClus}, we define $P_\textrm{suc}\left(K,\theta\right)=\Pr\left[\gamma>\theta,x_1>cd_K\right]
=\mathbb{E}\left[\Pr\{\gamma>\theta,x_1>cd_K|d_K\}\right]=\mathbb{E}\left[\Pr\left[\gamma>\theta|d_K\right]\Pr\left[x_1>cd_K|d_K\right]\right]
$, where the expectation is with respect to the distance $d_K$. Note that $\Pr\left[\gamma>\theta|d_K\right]$ follows from (\ref{deltaa}) with $\sigma^2\rightarrow 0$, while $\Pr\left[x_1>cd_K|d_K\right]=e^{-\lambda_u\pi c^2d_K^2}$. Following steps similar to those in Appendix C, and using the approximation $c=\sqrt{\frac{C_1(K)}{1+C_2(K)/\beta}}$ proposed in Section \ref{secClus}, we recover the expression in Theorem 3. 
\balance



\end{document}